%% file: submission.tex
\RequirePackage{fix-cm}
\documentclass[smallextended]{svjour3}       
\smartqed  
\usepackage{amsmath,amssymb}
\usepackage[utf8]{inputenc} 
\usepackage[dvipsnames]{xcolor}
\usepackage{todonotes}
\usepackage{bm}
\usepackage{tikz}
\usepackage{url}
\usetikzlibrary{decorations.pathreplacing}

\newif\ifdraft\draftfalse

\newif\ifjournal\journalfalse

\include{macros}

\include{macros-draft}

\begin{document}

\title{On long words avoiding Zimin patterns
}
\subtitle{}


\author{Arnaud Carayol \and
        Stefan G\"oller 
}


\institute{
A. Carayol \at
              Laboratoire d'informatique Gaspard Monge (UMR 8049) \\
Cit{\'e} Descartes \\
5, boulevard Descartes\\
Champs-sur-Marne \\
77454 MARNE-LA-VALLEE Cedex 2 \\
France\\
              \email{carayol@u-pem.fr}           
           \and
           S. G\"oller \at
	   Fachgebiet Theoretische Informatik / Komplexe Systeme\\
	   Fachbereich Elektrotechnik / Informatik\\
Universit\"at Kassel\\
Wilhelmsh\"oher Allee 73\\
34121 Kassel\\
Germany\\
\email{stefan.goeller@uni-kassel.de}           
}

\date{Received: date / Accepted: date}

\maketitle

\begin{abstract}
A pattern is encountered in a word if some infix of the word is the image of the pattern under some non-erasing morphism.
A pattern $p$ is unavoidable if, over every finite alphabet, every sufficiently long word encounters $p$.
A theorem by Zimin and independently by Bean, Ehrenfeucht and McNulty states that a pattern over $n$ distinct variables is 
unavoidable if, and only if, $p$ itself is encountered in the $n$-th Zimin pattern. Given an alphabet size $k$, we study the minimal length $f(n,k)$ such that every word of length $f(n,k)$
encounters the $n$-th Zimin pattern.
It is known that $f$ is upper-bounded by a tower of exponentials.
Our main result states that $f(n,k)$ is lower-bounded by a tower of $n-3$ exponentials, even for $k=2$. 
To the best of our knowledge, this improves upon a previously best-known doubly-exponential lower bound.
As a further result, we prove a doubly-exponential  upper bound for encountering Zimin patterns in the 
abelian sense.
\end{abstract}

\input{introduction}

\input{preliminaries}

\input{zimin-counters}

\input{binary}

\input{abelian}

\input{conclusion}



\bibliographystyle{plain}

\bibliography{biblio}


\end{document}

%% file: macros.tex
\newcommand{\cnt}[2]{[\![\,#1\,]\!]_{#2}}
\newcommand{\cntb}[2]{\{\!\!\{\,#1\,\}\!\!\}_{#2}}
\renewcommand{\max}[1]{{\bm{\tau}}(#1)}
\newcommand{\maxfun}{{\bm{\tau}}}

\newcommand{\ie}{\textrm{i.e.}}
\newcommand{\length}[1]{L_{#1}}
\newtheorem{fact}{Fact}
\newcommand{\Zimin}[1]{\mathrm{Zimin}(#1)}
\newcommand{\ZiminType}[1]{\mathrm{ZType}(#1)}
\newcommand{\intro}[1]{\emph{#1}}

\newcommand{\N}{\mathbb{N}}
\newcommand{\X}{\mathcal{X}}

%% file: macros-draft.tex
\newcommand{\modified}[1]{#1}
\ifdraft
\else

\newcommand\ac[1]{}
\newcommand\sg[1]{}
\newcommand\review[1]{}

\fi

%% file: introduction.tex
\section{Introduction}

\medskip

\newcommand{\Tower}{\mathrm{Tower}}
A pattern is a finite word over some set of pattern variables.
A pattern matches a word if the word can be obtain{ed} by substituting each 
variable appearing in the pattern by a non-empty word. 
The pattern $xx$ matches the word \emph{nana} when $x$ is replaced by the 
word \emph{na}.  
A word encounters a pattern if the pattern matches some infix of the word. 
For example, the word $banana$ encounters the pattern $xx$ 
(as the word \emph{nana} is one of its infixes). 
The pattern $xyx$ is encountered in precisely those words that contain 
two non-consecutive occurrences of the same letter, as e.g., the word $abca$.

A pattern is unavoidable if over every finite alphabet, every sufficiently long word encounters the pattern. Equivalently, by K\H{o}nig's Lemma, a pattern is unavoidable if over every finite alphabet all infinite words encounter the pattern. If it is not the case, the pattern is said to be avoidable.

The pattern $xyx$ is easily seen to be unavoidable since every sufficiently long word over a finite alphabet
must contain two non-consecutive occurrences of the same letter.
On the other hand, the pattern $xx$ is avoidable as Thue \cite{Thue06} gave an infinite word over a ternary alphabet that  does not encounter the pattern $xx$.

A precise characterization of unavoidable patterns was found by Zimin \cite{Zimin84} and
independently by Bean, Ehrenfeucht and McNulty \cite{BEM79}, see also for a more recent proof \cite{Sap95}. This characterization is based on a family $(Z_n)_{n\geq0}$ of unavoidable patterns, called the Zimin patterns, where
\[ 
   Z_1 = x_1  \quad\textrm{and}\quad Z_{n+1} = Z_n x_{n+1} Z_n \quad  \textrm{for all $n \geq 1$.}
\]
A pattern over $n$ distinct pattern variables is unavoidable if, and only if, 
the pattern itself is encountered in the $n$-th Zimin pattern $Z_n$.
Zimin patterns can therefore be viewed as the canonical patterns for unavoidability. 

Due to the canonical status of Zimin patterns it is natural to investigate 
\begin{quotation} 
``what is the smallest word length $f(n,k)$
that guarantees that every word over a $k$-letter alphabet of this length encounters the $n$-th 
pattern $Z_n$?''.
\end{quotation}

Computing the exact value of $f(n,k)$ for $n\geq 1$ and $k \geq 2$, or at least giving upper and lower bounds on its value, has been the topic of several articles in recent years \cite{CR14,Tao14,RS15,CR16}. 

For small values of $n$ and $k$, known results from \cite{RS15,Rora15} are 
summarized in the following table.
$$
\begin{array}{|c|c|c|c|c|c|}
\hline
 n  & 2 & 3 & 4 & 5 & k \\
 \hline  
 1 & 1 & 1 & 1 & 1 & 1 \\
 2 & 5 & 7 & 9 & 11 & 2k+1 \\
 3 & 29 & \leq 319 & \leq 3169 & \leq 37991 & \sqrt{e} 2^k (k+1)!+2k+1 \\
 4 & \in [10483,236489] & & & & \\ 
 \hline
 \end{array}
$$

In general, Cooper and Rorabaugh \cite[Theorem 1.1]{CR14} showed that the value 
of $f(n,k)$ is upper-bounded by a tower of exponentials of height $n-1$. 
To make this more precise let us define the tower function
$\Tower:\N\times\N\rightarrow\N$ inductively as follows:
$\Tower(0,k)=1$ and $\Tower(n+1,k)=k^{\Tower(n,k)}$ 
for all $n,k\in\N$.

\begin{theorem}[Cooper/Rorabaugh \cite{CR14}\label{T Cooper U}]
For all $n \geq 1$ and $k \geq 2$,
$f(n,k)\leq \Tower(n-1,K)$, where $K=2k+1$.
\end{theorem}

In stark contrast with this upper bound, Cooper and Rorabaugh showed that $f(n,k)$ is lower-bounded doubly-exponentially in $n$ for every fixed $k\geq 2$. 
To our knowledge, this is the best known lower bound for $f$.

\begin{theorem}[Cooper/Rorabaugh \cite{CR14}\label{T Cooper L}]
$f(n,k)\geq k^{2^{n-1}(1+o(1))}.$
\end{theorem}

This lower bound is obtained  by estimating the expected number of occurrences of $Z_n$ in long words over a $k$-letter alphabet using the first moment method.
%

\bigskip

\noindent
{\bf Our contributions. }
Our main contribution is to prove a lower bound for $f(n,k)$ that is non-elementary in $n$ even for $k=2$. We use Stockmeyer's yardstick construction \cite{Sto74} to construct for each $n \geq 1$, a family
of words of length at least $\Tower(n-1,2)$ (that we call higher-order counters here). 
We then show that a counter of order $n$ does not encounter $Z_n$ (for $n\geq 3$).
As these words are over an alphabet of size $2n-1$, this immediately establishes 
that
\[
  f(n,2n-1) {>} \Tower(n-1,2).
\] 
Stockmeyer's yardstick construction is a well-known technique to prove non-elementary lower bounds in computer science,
for instance it is used to show that the first-order theory of binary words with order 
\modified{is} non-elementary, see for instance \cite{Rein01} for a proof.

By using a carefully chosen encoding we are able to prove a lower bound for $f$ over a binary alphabet. Namely for all $n \geq 4$, it holds
\[
  f(n,2) {>} \Tower(n-3,2).
\] 

As a spin-off result, we also consider the abelian setting. Matching a pattern in the abelian sense is a weaker condition, where one only requires that {when an infix matches a pattern variable it} must only have the same number of occurrences of each letter
(instead of being the same words). This gives rise to the notion of 
avoidable in the abelian sense and unavoidable pattern in the abelian sense.
We note that every pattern that is unavoidable is in particular unavoidable in 
the abelian sense. 
However, the converse does not hold in general as witnessed by the pattern 
$xyzxyxuxyxzyx$, as shown in \cite{CuLi01}.
Even though Zimin patterns lose their canonical status in the abelian setting, 
the function $g(n,k)$, which is an abelian analog of the function $f(n,k)$,
has been studied \cite{Tao14}. 
For this function, Tao \cite{Tao14} establishes a lower bound that turns out to 
be  doubly-exponential from the estimations in \cite{JugePC}. 
The upper bound is inherited from the non-abelian setting and is 
hence non-elementary. 
We improve this upper bound to doubly-exponential. 
We also provide a simple proof using the first moment method that $g$ admits a
doubly-exponential lower bound which does not require the elaborate 
estimations of \cite{JugePC}.

{
{\bf Comparison with \cite{Colon17}.} While finalizing the present article, we became aware of the preprint \cite{Colon17} submitted in April 2017 in which Condon, Fox and Sudakov independently obtained non-elementary lower bounds for the function $f$. To keep the presentation clear, we present their contributions in this dedicated subsection.

Firstly the authors improve the upper-bound of $f$ of Theorem~\ref{T Cooper U} by showing \cite[Theorem~2.1]{Colon17} that for all $n \geq 3$, $k \geq 35$, 
\[
 f(n,k) \leq \Tower(n-1,k).
\]
They determine the value of $f(3,k)$ up to a multiplicative constant
and show \cite[Theorem~1.3]{Colon17} that:
\[
 f(3,k) = \Theta( 2^k k!).
\]

For the general case, they show \cite[Theorem 1.1]{Colon17} for any fixed $n \geq 3$,
\[
 f(n,k) \geq k
\underbrace{
  {{{^{k\vphantom{h}}}^{k\vphantom{h}}}^{\cdots\vphantom{h}}}^{k-o(k)\vphantom{h}}
}_{\text{$n-1$ times}}\]
They provide two proofs of this theorem. The first proof is based on  the probabilistic method and positively answer a question we ask in conclusion of our conference paper \cite{Carayol17b}. The second proof uses a counting argument. For the case of the binary of alphabet, they show \cite[Theorem 1.2]{Colon17} that:
\[
 f(n,2) \geq \Tower(n-4,2).
\]

In this last case, our bound is slightly better and has the extra advantage to provide a concrete \modified{family of words}
witnessing the bound.
}

\medskip

\noindent
{\bf Applications to the equivalence problem of deterministic pushdown automata.}
The equivalence problem for deterministic pushdown automata (dpda) is a famous problem
in theoretical computer science. 
Its decidability has been established by S\'{e}nizergues in 1997
and Stirling proved in 2001 the first complexity-theoretic upper bound, namely
a tower of exponentials of elementary height \cite{Stir02}
(in $\mathbf{F}_3$ in terms of Schmitz' classification \cite{Schmi13}), see
also \cite{Jan14} for a more recent presentation.

In \cite{Seni03} S\'{e}nizergues generalizes Stirling's approach by
a the so-called ``subwords lemma'' allowing him
both to prove a $\mathsf{coNP}$ upper bound for the equivalence
problem of finite-turn dpda and to explicitly link the complexity
of dpda equivalence with the growth of the function $f$: 
he shows that in case $f$ is elementary, 
then the complexity of dpda equivalence is elementary.

Inspired by this insight, a closer look reveals that the above-mentioned 
function $f$ has the same importance
in all complexity upper bound proofs \cite{Stir02,Seni03,Jan14}
for dpda equivalence.
However, due to Theorem \ref{T Main2} one cannot hope to improve the computational 
complexity of
dpda equivalence by proving an elementary upper bound on $f$ 
since $f(n,k)$ is shown to grow non-elementarily even for $k=2$
(Theorem \ref{T Main2}).

\medskip

\noindent
{\bf Organization of the paper.}
We introduce necessary {notations} in Section \ref{S Preliminaries}.
We show that $f(n,2n-1)\geq\Tower(2,n-1)$ in Section \ref{sec:Zimin-non-binary}.
We lift this result to unavoidability over a binary alphabet in
Section \ref{sec:zimin-binary}, where we show that 
$f(n,2)\geq\Tower(n-3,2)$ for all $n\geq 4$.
Our doubly-exponential bounds on abelian avoidability are presented
in Section \ref{S Abelian}.
We conclude in Section \ref{S Conclusion}.

%% file: preliminaries.tex
\section{Preliminaries}\label{S Preliminaries}



For every two integers $i,j$ we define $[i,j]=\{i,i+1,\ldots,j\}$ and
 $[j]=\{1,\ldots,{j}\}$.
By $\N$ we denote the non-negative integers and by $\N^+$ the positive integers.

If $A$ is a finite set of symbols, we denote by $A^*$ the set of all words 
over $A$ and by $A^+$ the set of all non-empty words over $A$. 
We write $\varepsilon$ for 
the empty word.  
For a word $u \in A^*$, we denote by $|u|$ its length. 
For two words $u$ and $v$, we denote by $u \cdot v$ (or simply $uv$) 
their concatenation.
A word $v$ is a \intro{prefix} of a word $u$, denoted by $v \sqsubseteq u$, if there exists a word $z$ such that $u=vz$. If $z$ is non-empty, we say that $v$ is a \intro{strict prefix}\footnote{Our definition of strict prefix is slightly non-standard as $\varepsilon$ is a strict prefix of any non-empty word. 
} of $u$. A word $v$ is a suffix of a word $u$ if there exist a word $z$ such that $u=zv$. If $z$ is non-empty, we say that $v$ is a \intro{strict suffix} of $u$. 

A word $v$ is an infix of a word $u$ if there exists $z_1$ and $z_2$ such that $u=z_1 v z_2$.
If both $z_1$ and $z_2$ are non-empty, $v$ is a \intro{strict infix\footnote{Again, remark that our definition is slightly non-standard as strict prefixes or strict suffixes are in general not  strict infixes.}} of $u$. If $v$ is an infix $u$ and $u$ can be written as $z_1{v}z_2$, the integer $|z_1|$ is called an \intro{occurrence} of $v$ in $u$. For $a \in A$, we denote by $|u|_a$ the number of occurrences of the symbol $a$ in $u$.

Given two non-empty sets $A$ and $B${,} a \emph{morphism} is a function $\psi:A^*\rightarrow B^*$ 
that satisfies {$\psi(uv)=\psi(u)\psi(v)$ for all $u,v\in A^*$}.
Thus, {every} morphism can simply {be given} by a function from $A$ to $B^*$.
A morphism $\psi$ is said to be \emph{non-erasing} if $\psi(a)\not=\varepsilon$ for all $a\in A$
and $\psi$ is \emph{alphabetic} if $\psi(a)\in B$ for all $a\in A$.

Let us fix a countable set $\X=\{x_1,x_2,\ldots\}$ of {\em pattern variables}.
A {\em pattern} is a finite word over $\X$. 
Let $\rho=\rho_1\cdots \rho_n$ be a pattern of length $n$.
A \modified{finite} word $w$ {\em matches} $\rho$ if $w=\psi(\rho)$ for some non-erasing morphism $\psi$.
A finite or infinite word $w$ {\em encounters} $\rho$ if some infix of $w$ matches $\rho$.

A pattern $\rho$ is said to be {\em unavoidable} if for all $k\geq 1$
all but finitely many finite words (equivalently every infinite word, by K\H{o}nig's Lemma) 
over the alphabet $[k]$ encounter $\rho$. 
Otherwise we say $\rho$ is \emph{avoidable}.

Unavoidable patterns are characterized by the so called Zimin patterns.

For all $n \geq 1$, the $n$-th {\em Zimin pattern $Z_n$} is given by:
\[
\begin{array}{lclcl}
Z_0 & = & \varepsilon & &  \\
Z_1 & = & x_1 & & \\
Z_{n+1} & = & Z_n x_{n+1} Z_n & & \textrm{for $n \geq 0$.}\\	
\end{array}
\]

For instance, we have $Z_1=x_1$, $Z_2=x_1x_2x_1$ and $Z_3=x_1x_2x_1 x_3 x_1 x_2 x_1$.

The following statement gives a decidable characterization of unavoidable patterns.
\begin{theorem}[Bean/Ehrenfeucht/McNulty \cite{BEM79}, Zimin \cite{Zimin84}]\label{T Zimin}
A pattern $\rho$ containing $n$ different variables is unavoidable if, and only if, 
$Z_n$ encounters $\rho$.
\end{theorem}

For instance, the pattern $x_1x_2x_1x_2$ is avoidable because it 
{is not encountered in $Z_2$ (not even in $Z_n$ for any $n\in\N$).}

Theorem \ref{T Zimin} justifies the study of the following Ramsey-like function.

\begin{definition}
Let $n,k\geq 1$. We define 
$$f(n,k)=\min\{\ell\geq 1\mid \forall w\in[k]^\ell: w\text{ encounters }Z_n\}.$$
\end{definition}

As we mainly work with Zimin patterns, we introduce the notions of Zimin type 
(\ie\ the maximal Zimin pattern that matches a word) and Zimin index 
(\ie\ the maximal Zimin pattern that a word encounters) and their basic properties.

\begin{definition}
	The Zimin type $\ZiminType{w}$ of a word $w$ is the largest $n$
	such that $w=\varphi(Z_n)$ for some non-erasing morphism $\varphi$.
\end{definition}

For instance, we have $\ZiminType{aaab}=1$, $\ZiminType{aba}=2$ and $\ZiminType{a^7ba^7}=4$. Remark that the Zimin type of any non-empty word is 
greater or equal to $1$ and the Zimin type of the empty word is $0$.

 Following the definition of Zimin patterns, the Zimin type of a word can be inductively characterized as follows:
\begin{fact}
\label{fact:ZiminType-inductive}
For any non-empty word $w$, 
\[
\ZiminType{w}=1+\mathrm{max}\{ \ZiminType{\alpha} \mid w=\alpha \beta \alpha\;\textrm{for non-empty $\alpha$ and $\beta$}\},
\]
with the convention that the maximum of the empty set is $0$.	
\end{fact}

\begin{definition}
	The Zimin index $\Zimin{w}$ of a non-empty word $w$ the maximum Zimin type of an infix of $w$.
\end{definition}

For instance, we have $\Zimin{aaab}=2$ {and} $\Zimin{bbaba}=2$. 
As a further example note that $\Zimin{baaabaaa}=3$ although $\ZiminType{baaabaaa}=1$.

\begin{lemma}
\label{lemma:zimin-monotone}
For any  word $w$, we have the following properties:
\begin{itemize}
	\item $\ZiminType{w} \leq \Zimin{w}$,
	\item for any infix $w'$ of $w$, $\Zimin{w'} \leq \Zimin{w}$,	
	\item $\Zimin{w} \leq \lfloor \log_2(|w|+1) \rfloor$.
\end{itemize}
\end{lemma}

\begin{proof}
The first two points directly follow from the definition. For the last point, remark that 
for a word $w$ to encounter the $n$-th Zimin pattern $Z_n$, it must be of
length a least $|Z_n|$. 
As $Z_n$ has length $2^n-1$, we have
\[ 2^{\Zimin{w}}-1 \leq |w|,\]
which implies the announced bound. 
\end{proof}


%% file: zimin-counters.tex
\section{The Zimin index of higher-order counters}
\label{sec:Zimin-non-binary}

In this section we show that there is a family of words, that we refer to as ``higher-order counters'',  whose length is non-elementary in $n$ and whose Zimin index
is $n-1$, allowing us to show that $f(2n-1,n){> \Tower(n-1,2)}$.
In Section \ref{S Stockmeyer} we introduce higher-order counters and in Section \ref{S Small} we show that their Zimin index is precisely $n-1$ including the mentioned
lower bound on $f$.

\subsection{Higher-order counters {\`a} la Stockmeyer}\label{S Stockmeyer}

In this section we introduce counters that encode values range from $0$ to a tower of exponentials.
To the best of our knowledge this construction {was} introduced by Stockmeyer to show non-elementary complexity lower bounds and is often referred to as the 
``yardstick construction'' \cite{Sto74}.
We refer to such counters as ``higher-order counters'' in the following.

We define the {(unary)} {\em tower function} $\maxfun:\N\rightarrow\N$ as 
\[
\begin{array}{llll}
	\max{0} &=& 1 &\quad\text{ and}\\
	\max{n+1} & = &2^{\max{n}} &\quad \text{ for all $n \geq 0$.}
\end{array}
\]
{Equivalently, $\max{n}=\Tower(n,2)$ for all $n\in\N$.}
For all $n \geq 1$, we define an alphabet $\Sigma_n$
by taking $\Sigma_1 = \{0_1,1_1\}$ and for all $n>1$, $\Sigma_{n} = \Sigma_{n-1} \cup \{0_n,1_n\}$. We say the symbols $0_n$ and $1_n$ have {\em order $n$}.
We define $\Sigma = \cup_{n \geq 1} \Sigma_n$ to be {the} set of all these symbols.

For all $n \geq 1$ and for all $i \in [0,\max{n}-1]$,
we define a word over $\Sigma_n$ called \emph{the $i$-th counter of order $n$} and denoted by $\cnt{i}{n}$.
The definition proceeds by induction on $n$. For $n=1$,
there are only two counters $\cnt{0}{1}$ and $\cnt{1}{1}$ (recall that $\max{1}=2$). We define
\[
\cnt{0}{1} = 0_1 \;\text{and}\; \cnt{1}{1} = 1_1.
\]
For $n\geq 1$ and $i \in [0,\max{n+1}-1]$ we define
\[
 \cnt{i}{n+1}= \cnt{0}{n} b_0 \cnt{1}{n} b_1 \cdots \cnt{\max{n}-1}{n} b_{\max{n}-1},
\]
where $b_0 b_1 \cdots b_{\max{n}-2}b_{\max{n}-1}$ is the binary decomposition of $i$ over the alphabet $\{0_{n+1},1_{n+1}\}$ with $b_0$ the least significant bit 
(\ie\ $i = \sum_{j=0}^{\max{n}-1} \overline{b_j}\cdot 2^j$ where $\overline{b_j}=0$ if $b_j=0_{n+1}$ and $\overline{b_j}=1$ if $b_j=1_{n+1}$). 

\noindent
For instance, there are $\tau(2)=4$ counters of order $2$, namely
\[
\begin{array}{lclclcl}
\cnt{0}{2} & = & 0_1 \mathbf{0_2} 1_1 \mathbf{0_2}, & & \cnt{1}{2} &=& 0_1 \mathbf{1_2} 1_1 \mathbf{0_2}, \\
\cnt{2}{2} & = & 0_1 \mathbf{0_2} 1_1 \mathbf{1_2}, & &
\cnt{3}{2} & =&  0_1 \mathbf{1_2} 1_1 \mathbf{1_2}. \\
\end{array}
\]
\noindent
For $\cnt{11}{3}$, we have $11= 1\cdot 2^0 + 1\cdot 2^1 + 0\cdot 2^2 + 1 \cdot 2^3$ and hence 
\[
\cnt{11}{3}= \underbrace{0_1 0_2 1_1 0_2}_{\cnt{0}{2}} \mathbf{1_3} \underbrace{0_1 1_2 1_1 0_2}_{\cnt{1}{2}} \mathbf{1_3} \underbrace{0_1 0_2 1_1 1_2}_{\cnt{2}{2}} \mathbf{0_3} \underbrace{0_1 1_2 1_1 1_2}_{\cnt{3}{2}} \mathbf{1_3}.
\]
\noindent
The following lemma {is} easily be proven by induction on $n$.

\begin{lemma}
\label{lemma:basics-counters}
Let $n\geq 1$.
\begin{enumerate}
\item There are $\max{n}$ counters of order $n$.
\item For all $i \neq j \in [0,\max{n}-1]$ we have $\cnt{i}{n} \neq \cnt{j}{n}$. 
\item 
If $n>1$, then for all $i \in [0,\max{n}-1]$ and 
$j \in [0,\max{n-1}-1]$ the counter $\cnt{j}{n-1}$ has exactly one occurence in $\cnt{i}{n}$.
\end{enumerate}
\end{lemma}

{The following lemma expresses that the order of a symbol in a counter of order $n$ only depends on the distance of this symbol to an order $n$ symbol.} It  
 {is} proven by induction on $n$ by making use of the previous lemma.

\begin{lemma}
\label{lem:order-from-position}
Let $n \geq 2$, $i \in [0,\max{n}-1]$ and $p,p'$ and $\ell$ such that ${p,p', }p + \ell$ and $p' + \ell$ belong to  $[0,|\cnt{i}{n}|-1]$.
If the symbols occurring at $p+\ell$ and $p'+\ell$ in 
$\cnt{i}{n}$ are of order $n$ then the symbols occurring at $p$ and $p'$ {in $\cnt{i}{n}$} have the same order.
\end{lemma}
\noindent
The length of an order-$n$ counter, denoted by $\length{n}$ is inductively defined as follows:
\[
\begin{array}{lclcl}
\length{1} & = & 1 & & \\
\length{n+1} & = & \max{n} \cdot (\length{n} + 1)& &\quad \textrm{for all $n \geq 1$} \\
\end{array}
\]

\noindent
Note that in particular for all $n \geq 1$ we have $\length{n} \geq \max{n-1}$.

\subsection{Higher-order counters have small Zimin index}\label{S Small}

The aim of this section is to give an upper bound on the Zimin index of counters of order $n$. 
A first simple remark is that the Zimin index of any counter of order $n$ is 
{upper-}bounded by the index of $\cnt{0}{n}$.

\begin{lemma}
\label{lem:zero-upper-bound}
For all $n\geq 1$ and for all $i \in [0,\max{n}-1]$,
\[
\Zimin{\cnt{i}{n}} \leq \Zimin{\cnt{0}{n}} .
\]
\end{lemma}
\begin{proof}
Let $n\geq 1$ and let $i \in [0,\max{n}-1]$. By definition of higher-order counters, we have
\begin{equation}
\label{eq:projection}
 \cnt{0}{n} = \psi(\cnt{i}{n}).
\end{equation}
where $\psi$ is the alphabetic morphism defined by $\psi(0_n)=\psi(1_n)=0_n$ and $\psi(x)=x$ for all $x \in \Sigma_{n-1}$. Assume that $\cnt{i}{n}$ contains an infix of the form $\varphi(Z_\ell)$ for some non-erasing morphism $\varphi$ and $\ell \geq 0$. By Eq.~\ref{eq:projection}, $\cnt{0}{n}$ contains $\psi(\varphi(Z_\ell))$ as an infix. It follows that $\Zimin{\cnt{0}{}} \geq \Zimin{\cnt{i}{n}}$.
\end{proof}

This leads us to the main result of this section.

\begin{theorem}
\label{thm:zimin-non-binary-counters}
For all $n\geq 3$,
\[
\Zimin{\cnt{0}{n}} \leq  n-1 {.}
\]
\end{theorem}

\begin{proof}
 The proof proceeds by induction on $n\geq 3$. 
For the base case, we have to show that 
 \[
 \cnt{0}{3} = 0_1 0_2 1_1 0_2 \mathbf{0_3} 0_1 1_2 1_1 0_2 \mathbf{0_3} 0_1 0_2 1_1 1_2 \mathbf{0_3} 0_1 1_2 1_1 1_2 \mathbf{0_3}
 \]
has a Zimin index of at most 2. 
Assume towards a contradiction that $\cnt{0}{3}$ has Zimin index at least 3. Then it must contain an infix of the form $\alpha \beta \alpha$ for some non-empty $\alpha$ and $\beta$ with $\alpha$ of Zimin type at least 2. In particular, $\alpha$ must be of length at least 3. 
A careful inspection shows that the only infixes of $\cnt{0}{3}$ of length at least $3$ that appear twice are the following:
\[
\begin{array}{ccccccc}
0_1 0_2 1_1 & &
1_1 0_2 \mathbf{0_3} &&
1_1 0_2 \mathbf{0_3} 0_1 &&
0_2 \mathbf{0_3} 0_1 \\
\mathbf{0_3} 0_1 1_2 && 
\mathbf{0_3} 0_1 1_2 1_1 &&
0_1 1_2 1_1 && 
1_1 1_2 \mathbf{0_3} \\
\end{array}
\]
All these words have Zimin type 1 which concludes the base case.

	 Assume that the property holds for some $n \geq 3$. By Lemma~\ref{lem:zero-upper-bound} 
	 \modified{and induction hypothesis}, we have that for all $i \in [0,\max{n}-1]$,
\begin{equation}
\label{eq:induction-hypothesis}
	\Zimin{\cnt{i}{n}} \leq n-1.
\end{equation}   

Let us show that $\Zimin{\cnt{0}{n+1}}\leq n$. Let $\alpha \beta \alpha$ be an infix of $\cnt{0}{n+1}$ for some non-empty words $\alpha$ and $\beta$. It is enough to show that $\ZiminType{\alpha} \leq n-1$. We distinguish the following cases depending on the number occurrences of $0_{n+1}$ in $\alpha$.

\paragraph*{Case 1: $\alpha$ contains no occurrence of $0_{n+1}$.} Then $\alpha$ is an infix of some $\cnt{i}{n}$ for some $i \in [0,\max{n}-1]$. By induction hypothesis (\ie\  Eq.~{\ref{eq:induction-hypothesis}}) and Lemma~\ref{lemma:zimin-monotone}, $\ZiminType{\alpha}\leq\Zimin{\alpha}\leq\Zimin{\cnt{i}{n}}\leq n-1$.

\paragraph*{Case 2: $\alpha$ contains at least two occurrences of $0_{n+1}$.}
	By definition of counters, $\alpha$ has an infix $0_{n+1} \cnt{i}{n} 0_{n+1}$  for some $i \in [0,\max{n}-1]$. Hence $\cnt{0}{n+1}$ would contain two occurrences of $0_{n+1}\cnt{i}{n}0_{n+1}${,} which {contradicts Lemma~\ref{lemma:basics-counters}\modified{(3)}}.

\paragraph*{Case 3: $\alpha$ contains exactly one occurrence of $0_{n+1}$.} 
By definition of $\cnt{0}{n+1}$, there exists $i \neq j \in [0,\max{n}-1]$ such that $\alpha$ is of the form 
$u 0_{n+1} v$ with $u$ a suffix of both  $\cnt{i}{n}$ and $\cnt{j}{n}$ and $v$ a prefix of 
both $\cnt{i+1}{n}$ and $\cnt{j+1}{n}$.

	Consider the morphism $\psi$ \modified{that} erases all symbols in $\Sigma_{n-1}$ and replaces $0_n$ 
and $1_n$ by $0$ and $1$, respectively. Let us assume that
\[
\begin{array}{lcl}
	\psi(u) & = & b_{\max{n-1}-\ell_0} \cdots b_{\max{n-1}-1} \\
	\psi(v) &=  & c_0 \cdots c_{\ell_1-1} \\
\end{array}
\]
for some $\ell_0 \in [0,\max{n-1}]$ and $\ell_1 \in [0,\max{n-1}]$ and $b_k \in \{0,1\}$ for all $k \in [\max{n-1}-\ell_0, \max{n-1}-1]$ and $c_k \in \{0,1\}$ 
for all $k \in [0,\ell_1-1]$.

Let us start by showing that
\begin{equation}
	\ell_0 + \ell_1 < \max{n-1}.
\end{equation} 

By definition of counters,  
$b_{\max{n-1}-1}$ is the most significant bit of the binary presentation (of length $\max{n-1}$) of $i$ and $j$  and $c_0$ is the least significant bit of the binary presentation
of both $i+1$ and $j+1$. More formally, there exist ${x_i,x_j} \in [0,2^{\max{n-1}-\ell_0}-1]$ and  
${y_i,y_j} \in [0,2^{\max{n-1}-\ell_1}-1]$ such that:

\[
\begin{array}{lclclcl}
 i   & = & x_i + 2^{\max{n-1}-\ell_0} \cdot B  &  \quad &
 j   & = & x_j + 2^{\max{n-1}-\ell_0} \cdot B  \\ 
 i+1 & = & C + 2^{\ell_1} y_i  & \quad & 
 j+1 & = & C + 2^{\ell_1} y_j \\
\end{array}
\]
\noindent with
\[
\begin{array}{lclclcl}
B & = & \sum\limits_{k=0}^{\ell_0-1} b_{\max{n-1}-\ell_0-k} \cdot 2^k	 & \quad &
C & = & \sum\limits_{k=0}^{\ell_1-1} c_k \cdot 2^k.  \\
\end{array}
\]
\noindent
Assume towards a contradiction that $\ell_0+\ell_1 \geq\max{n-1}$. 
In particular, this implies $ 2^{\ell_1} \geq  2^{\max{n-1}-\ell_0}$. And hence,

\[ 
\begin{array}{lclcl}
x_i & = & i \mod 2^{\max{n-1}-\ell_0}    & & \textrm{by definition of $i$}\\
    & = & C-1 + 2^{\ell_1} y_i \mod 2^{\max{n-1}-\ell_0}	& & \\
    & = & C-1 \mod 2^{\max{n-1}-\ell_0} & & \textrm{as $2^{\max{n-1}-\ell_0}$ divides $2^{\ell_1}$.}
\end{array}
\]

A similar reasoning shows that $x_j= C-1 \mod 2^{\max{n-1}-\ell_0}$. 
Hence $x_i=x_j$ and hence $i=j$ which brings the contradiction.

\modified{Having just shown} $\ell_0+\ell_1 < \max{n-1}$, there exists some $i_0 \in [0,\max{n}-1]$ such that
$v$ is a prefix and $u$ is a suffix of $\cnt{i_0}{n}$. 
That is, the binary representation of $i_0$ of length $\max{n-1}$ has 
$c_0 \cdots c_{\ell_1-1}$ as $\ell_1$ least significant bits and  
$b_{\max{n-1}-\ell_0}\cdots b_{\max{n-1}-1}$ as $\ell_0$ most significant bits. 
In particular, as $\ell_0 + \ell_1 < \max{n-1}$, we have that:
\begin{equation}
\label{eq:decomposition_i_0}
\cnt{i_0}{n} = v r u \quad \textrm{for some non-empty $r$}.\end{equation}

We claim that $\ZiminType{\alpha} \leq \Zimin{\cnt{i_0}{n}}$ by which we would be done
since then $\ZiminType{\alpha}\leq\Zimin{\cnt{i_0}{n}}\leq n-1$
by Lemma \ref{lem:zero-upper-bound} and induction hypothesis.

Assume that $\alpha=\gamma \delta \gamma$ for non-empty $\gamma$ and $\delta$.
Using Fact~\ref{fact:ZiminType-inductive}, it is enough to show that $\ZiminType{\gamma}+1 \leq \Zimin{\cnt{i_0}{n}}$.  Recall that $\alpha=u 0_{n+1} v$ and $\alpha$ contains only one occurrence of $0_{n+1}$. It follows {that} $\gamma$ must be a prefix of $u$ and a suffix of $v$. In particular using Eq.~\ref{eq:decomposition_i_0},  $\cnt{i_0}{n}$ contains $\gamma r \gamma$ as an infix.
\[
\begin{array}{lcl}
 \ZiminType{\gamma} +1 & \leq & \ZiminType{\gamma r \gamma} \\
                      & \leq &  \Zimin{\cnt{i_0}{n}} \\

 \end{array}
\]

\end{proof}


The upper bound on the Zimin index of higher-order counters established in the previous theorem is tight.

\begin{theorem}
\label{thm:zimin-non-binary-counters-exact}
\[
\begin{array}{lclcll}
\Zimin{\cnt{0}{1}} & = & \Zimin{\cnt{0}{1}} & = & 1  \\
\Zimin{\cnt{0}{2}} & = & \Zimin{\cnt{3}{2}} & = & 2 \\
\Zimin{\cnt{1}{2}} & = & \Zimin{\cnt{2}{2}} & = & 1 \\
	\modified{\Zimin{\cnt{i}{3}}} & = & 2 &  & &\quad \textrm{for all $i\in[0,\tau(3)-1]$}\\
	\Zimin{\cnt{i}{n}} & = & n-1  &  &&\quad \textrm{for all $\modified{n \geq 4}$ and all $i \in [0,\max{n}-1]$} \\
\end{array}
\]
\end{theorem}

\begin{proof} 
	\modified{The Zimin index of all higher-order counters of order at most $3$ is checked by a computer program.}
	For the last statement, Theorem~\ref{thm:zimin-non-binary-counters} established that $\Zimin{\cnt{i}{n}} \leq n-1$ for \modified{all} $n \geq 3$ and \modified{all} $i \in [0,\max{n}-1]$. 
	\modified{ We need to prove $\Zimin{\cnt{i}{n}}\geq n-1$ for all $n\geq 4$.}
	\modified{That is, we} only need to show that \modified{for all $n\geq 4$,} $\cnt{i}{n}$ contains an infix of Zimin type at least $n-1$. We prove the stronger property that for all $n \geq 3$, there exists
a word $\alpha_n \in \Sigma^*$ of Zimin type at least $n-1$,
	which is an infix of both $\cnt{2}{n}$ and $\cnt{3}{n}$ 
	\modified{(recall that in particular when $n\geq 3$, every higher-order counter
	$\cnt{i}{n+1}$ contains both $\cnt{2}{n}$ and $\cnt{3}{n}$ as infix)}.

We proceed by induction on $n \geq 3$.

For the base case $n=3$, we take $\alpha_3=1_1 1_2 1_1$. Clearly $\ZiminType{\alpha_3}=2$ and as $\alpha_3$ is an infix of $\cnt{3}{2}$, it is also an infix of $\cnt{2}{3}$ and $\cnt{3}{3}$ (and in fact of all order 3 counters).

Assume that the property holds for some $n \geq 3$ and let us show that it holds for $n+1$. From the definitions, we have:
\[
\begin{array}{lcl}
 \cnt{2}{n+1}  & =  & \cnt{0}{n} \mathbf{0_{n+1}} \cnt{1}{n} \mathbf{1_{n+1}} \cnt{2}{n} \mathbf{0_{n+1}} \cnt{3}{n} \cdots \\
 \cnt{3}{n+1}  & =  & \cnt{0}{n} \mathbf{1_{n+1}} \cnt{1}{n} \mathbf{1_{n+1}} \cnt{2}{n} \mathbf{0_{n+1}} \cnt{3}{n} \cdots \\
\end{array}
\]
In particular, $\cnt{2}{n}0_{n+1}\cnt{3}{n}$ is an infix of both $\cnt{2}{n+1}$ and $\cnt{3}{n+1}$. By induction hypothesis, $\alpha_n$ is an infix of $\cnt{2}{n}$ and of $\cnt{3}{n}$. Therefore there exist $x_2,y_2,x_3$ and $y_3$ such that:
\[
\begin{array}{lcl}
\cnt{2}{n} & =& x_2 \alpha_n y_2 \\
\cnt{3}{n} & =& x_3 \alpha_n y_3 \\
\end{array}
\]
In particular, the common infix can be written as:
\[
\cnt{2}{n}0_{n+1}\cnt{3}{n} = x_2 \alpha_n y_2 0_{n+1} x_3 \alpha_n y_3 \\
\]

We can take $\alpha_{n+1}=\alpha_n y_2 0_{n+1} x_3 \alpha_n$ which is therefore an infix of both  $\cnt{2}{n+1}$ and $\cnt{3}{n+1}$. The Zimin type of $\alpha_{n+1} = \alpha_n \beta \alpha_n$ with $\beta=y_2 0_{n+1} x_3$ is at least 1 more than that of $\alpha_n$. By induction hypothesis, the Zimin type of $\alpha_{n+1}$ is therefore at least $n$ which concludes the proof.
 \end{proof}

\begin{corollary}
	For all $n \geq 3$,
	\[
	   f(n,2n-1) {>} \max{n-1}.
	\]
\end{corollary}

\begin{proof}
Let $n \geq 3$. The word $\cnt{0}{n}$ over the alphabet 
$\{0_1,1_1,\ldots,0_{n-1},1_{n-1},0_n\}$ of size $2n-1$ has length at least $L_n\geq\max{n-1}$ and Zimin index at most $n-1$. This word hence avoids the $n$-th Zimin pattern and 
witnesses the announced lower bound.
\end{proof}

%% file: binary.tex
\section{Reduction to the binary alphabet}
\label{sec:zimin-binary}

In this section, we show how to encode a higher-order counter seen in 
Section~\ref{sec:Zimin-non-binary} over the binary alphabet $\{0,1\}$ while still preserving a relatively
{low} upper bound on the Zimin index. For this we apply to counters the morphism $\psi$\modified{,} defined as follows
\[
	\psi(0_n)  =  00\,(01)^{n-1}\, 00 \quad \quad 	\psi(1_n)  =  11 \,(01)^{n-1} \,11 \quad \quad \quad \textrm{for all $n \geq 1$}.
\]

\begin{definition}
	For all $n \geq 1$ and $i \in [0,\max{n}-1]$, we define
	\[
	\cntb{i}{n} = \psi(\cnt{i}{n}).
	\]
\end{definition}

The set of images of the letters by this morphism forms what is known as an infix code, 
i.e. $\psi(a)$ is not an infix of $\psi(b)$ for any two letters $a,b\in\Sigma$ with $a\not=b$. In addition to being an infix code, the morphism was designed so that:
\begin{itemize}
\item we are able to attribute a non-ambiguous partial decoding to most infixes of an encoded word (cf. Lemma~\ref{lem:unique-parse}),
\item the encoding of $0_n$ and $1_n$ differ on their first and last symbol \modified{inter alia},
\item the Zimin index of the encoding of an order $n$ symbol is relatively low (we will show that it is at most $\lfloor \log_2(3+2n) \rfloor$).
\end{itemize}

Applying a non-erasing morphism to a word can only increase its Zimin index. 
We will see in  the remainder of this section that the Zimin index of higher-order counters is increased by at most 2 when the morphism  $\psi$ is applied.  It is possible that another choice of morphism would bring a better upper bound. However, remark 
that the proof we present is tightly linked to the above-mentioned properties of $\psi$ that are decisive for the construction to work.  
 
This section is devoted to establishing the following theorem.  
 
 \begin{theorem}\label{T binary}
 For all $n \geq 2$ and for all $i \in [0,\max{n}-1]$,
 \[
  \Zimin{\cntb{i}{n}}\leq n+1.
 \]	
 \end{theorem}

The proof{,} which is essentially {an} {extension of the} proof of 
Theorem~\ref{thm:zimin-non-binary-counters}{,} is given {in} Section~\ref{ssec:proof-non-binary}.  To perform this reduction, we first establish basic properties of the morphism $\psi$ and its decoding in Section~\ref{ssec:parsing}.

Recalling that an order $n$ counter {has} length at least $\max{n-1}$, in particular {so} does its
code, which is its image under the non-erasing morphism $\psi$. 

{Theorem \ref{T binary} immediately
 implies} {the following} non-elementary lower bound for $f(n,2)$ whenever $n \geq 4$.

\begin{theorem}{\label{T Main2}} For all $n \geq 4$, 
\[
f(n,2){>} \max{n-3}.
\]
\end{theorem}

\subsection{Parsing the code $\psi$}
\label{ssec:parsing}

A word $w \in \{0,1\}^*$ is {\em coded} by $\psi$ (or simply a {\em coded word}) if it is the image by $\psi$ of some word $v$ over $\Sigma^*$. As the image of $\psi$ is an infix code, the word $v \in \Sigma^*$ is unique. However in our proof, we need to take into consideration all infixes of a coded word. 
To be able to reuse the proof techniques of Theorem~\ref{thm:zimin-non-binary-counters}, it is necessary to 
associate to an infix of a coded word a partial decoding called a \intro{parse}. 

Let us therefore consider the following sets,
\begin{itemize}
	\item $C = \psi(\Sigma)= \{ \psi(0_k) \mid k \geq 1\} \cup \{ \psi(1_k) \mid k \geq 1\} $ the image of the morphism {of letters in $\Sigma$,}
	\item $L = \{ v\in\{0,1\}^*  \mid \exists u \in \{0,1\}^+, uv \in C\}$ the set of strict suffixes of $C$,
	\item  $R = \{ u \in \{0,1\}^* \mid \exists v \in \{0,1\}^+, uv \in C\}$ the set of strict prefixes of $C$, and
	\item $F=\{ v \in \{0,1\}^* \mid \exists u,w \in \{0,1\}^+, uvw \in C\}$ the set of strict infixes of $C$.
\end{itemize}

\noindent
Remark that these four sets are all regular languages. Indeed,
\begin{itemize}
	\item $C=00 (01)^* 00 \cup 11 (01)^* 11$,
	\item $L=\varepsilon+0+1+ (\varepsilon+1)(01)^*11 + (\varepsilon+0+1)(01)^*00$,
	\item $R=\varepsilon+0+1+ 11(01)^*(\varepsilon+0+1) + 00(01)^*(\varepsilon+0)$, and
	\item \modified{$F=(\varepsilon +0)(01)^*(\varepsilon+0)+(\varepsilon+1)(01)^*(\varepsilon+0+1)$.}
\end{itemize}


\modified{Thanks to this property, it is possible to reduce the proofs of our statements (namely Lemma~\ref{lem:observations} and Lemma~\ref{lem:unique-parse}) to computations on finite word automata and finite word transducers. For our tests, we used Awali\footnote{\url{http://vaucanson-project.org/AWALI/Awali.html}} which is the latest version of the well-established Vaucanson platform.}

{We collect in the following lemma some observations on these sets which will be used through out the proofs in this section.

\begin{lemma}
\label{lem:observations}
The sets $C,R,L$ and $F$ satisfy the following equations:
\begin{enumerate}
	\item  \modified{$L C^* R \cap F=L R \cap F=\{0,1\}^{\leq 2} \cup \{001,110\}$},
	\item  $LR \cap C = \{0000,1111\}$,
\end{enumerate}
	
\end{lemma}
}

\begin{lemma}
\label{lem:characterization-infix-coded-word}
\modified{
	For every $n\geq 1$ and every infix of a word in $C_{\leq n}^*$  belongs to \[
F_{\leq n} \cup L_{\leq n}C^*_{\leq n}R_{\leq n}\]
where $F_{\leq n}$, $C_{\leq n}$ and $R_{\leq n}$ respectively denote the restrictions of $F$, $C$ and $R$ to symbols of order at most $n$. }
\end{lemma}
\begin{proof}
	\modified{Let us first remark that $\varepsilon\in L_{\leq n}$ and $\varepsilon\in R_{\leq n}$.}
We show, by induction on $m$, that every infix $\alpha$ of a  word in $C_{\leq n}^m$ belongs to $F_{\leq n}\cup L_{\leq n}C_{\leq n}^*R_{\leq n}$.
	For $m=0$, the property is immediate as $\varepsilon \in \modified{F_{\leq n}}$ (and to $L_{\leq n}C_{\leq n}^*R_{\leq n}$). 
For $m=1$, if  $\alpha$ is an infix of $c \in C_{\leq n}$, then $c=x\alpha y$ for some words $x$ and $y$. 
If $x \neq \varepsilon$ and $y \neq \varepsilon$ then $\alpha$ belongs to $F_{\leq n}$. 
If $x \neq \varepsilon$ and $y = \varepsilon$ then $\alpha$ belongs $L_{\leq n}$. 
	If $x = \varepsilon$ and $y \neq \varepsilon$ then $\alpha$ belongs \modified{to $R_{\leq n}$}. 
	Finally\modified{,} if $x=y=\varepsilon$, then $\alpha = c \in C_{\leq n}\subseteq L_{\leq n}C_{\leq n}^*R_{\leq n}$.
For the induction step, let $\alpha$ be an infix of some $c_1 \cdots c_{m+1} \in C_{\leq n}^{m+1}$, where $m\geq 1$.
There are the following cases: either $\alpha$ is an infix of $c_1 \cdots c_{m}$  and we can conclude using the 
	induction hypothesis, or $\alpha$ is an infix of \modified{$c_{m+1}$} 
	and we can conclude using the case \modified{$m=1$}. Finally it remains the case where $\alpha=xy$ with $x$ a suffix of $c_1 \cdots \modified{c_m}$ and $y$ a prefix of \modified{$c_{m+1}$}.
Clearly $x$ belongs to $L_{\leq n} C_{\leq n}^*$ and $y$ belongs to $R_{\leq n} \cup C_{\leq n}$. This implies that 
$\alpha$ belongs to $L_{\leq n} C_{\leq n}^* R_{\leq n}$.
\end{proof}

This leads us to define a \intro{parse}  $p$ as a triple $(\ell,u,r)$ in $L \times \Sigma^* \times R$.  The word $u$ will be called the \intro{center} of the parse $p$.
The \intro{value} of the parse $(\ell,u,r)$ is the word $\ell \psi(u) r \in\{0,1\}^*$. 
We say that \intro{$\alpha$ admits $p$} if $\alpha$ is the value of $p$. 

By the above fact, for all coded words all of its infixes not belonging to $F$
have at least one parse. 
However, the parse is not necessarily unique. For instance, consider the infix  $\alpha=0000000000=0^{10}$ which appears in $\psi(0_1 0_1 0_1)$. It can be parsed as
\[
(\varepsilon,0_1 0_1,00) \;,\; (0,0_1 0_1,0) \;\textrm{ and as } \; (00,0_1 0_1,\varepsilon).
\] 

However, we will provide sufficient conditions on an infix to admit a unique parse. 

\begin{definition}
\label{def:simple}
	A word $\alpha \in\{0,1\}^*$ is \intro{simple} if either $|\alpha| < 11$, or $\alpha$ belongs to $F$ or $0^{10}$ is an infix of $\alpha$ or $1^{10}$ is an infix of $\alpha$.
\end{definition}

On the one hand, the term {``simple''} is justified in the context of this proof as simple infixes of $\cntb{i}{n}$ can 
\modified{rather easily shown to have} Zimin index at most $n-1$ for all $n \geq 4$ and all
$i\in[0,\max{n}-1]$ (cf. Lemma~\ref{lem:Zimin-simple}). 

\begin{lemma}
\label{lem:Zimin-simple}
	For all $n \geq 4$ \modified{and} all $i \in [0,\max{n}-1]$ \modified{every} simple infix of $\cntb{i}{n}$ has Zimin index at most $n-1$.
\end{lemma}

\begin{proof}
Let $n \geq 4$, let $i \in [0,\max{n}-1]$ and let $\alpha$ be a simple infix of $\cntb{i}{n}$.
It is easy to check that $\cnt{i}{n}$ does not contain two consecutive occurrences of $0_1$ or of $1_1$. 
It follows that $\cntb{i}{n}$ cannot contain more than 9 consecutive zeros or 9 consecutive ones. 
Hence $\alpha$ is simple either because it belongs to $F$ or because $|\alpha|<11$.

First consider the case where $|\alpha|<11$. We know, by Lemma~\ref{lemma:zimin-monotone}, that $\Zimin{\alpha} \leq \lfloor \log_2(11) \rfloor = 3 \leq n-1$.

\modified{
	Now consider the case where $\alpha$ belongs to $F$.  By Lemma~\ref{lem:characterization-infix-coded-word}, we have that $\alpha$ \modified{either} belongs to $F \cap F_{\leq n}= F_{\leq n}$ or 
to $F \cap L_{\leq n} C_{\leq m}^* R_{\leq n}$.

In the first case, $|\alpha|\leq 2n+2$ and {thus }by Lemma~\ref{lemma:zimin-monotone}, 
	$\Zimin{\alpha} \leq \lfloor \log_2(2n+\modified{2}) \rfloor \leq n-1$ {since} $n \geq 4$.

In the second case $\alpha$ belongs to $F \cap L_{\leq n} C_{\leq m}^* R_{\leq n} \subseteq L C^* R \cap F$. From Lemma~\ref{lem:observations}, we have $L C^* R \cap F=L R \cap F=\{0,1\}^{\leq 2} \cup \{001,100\}$. 
Hence in this case, we have $\Zimin{\alpha}\leq 1$.

  }

\end{proof}

{On the other hand, the term ``simple'' will be justified by the fact that for non-simple infixes there 
is exactly one possible parse as shown in the following lemma.}

\begin{lemma}
\label{lem:unique-parse}
Any non-simple infix of a coded word admits a unique parse.
\end{lemma}

\begin{proof}
\ifjournal
\else
The existence of the parse is immediate from Lemma~\ref{lem:characterization-infix-coded-word}. 
To prove unicity of the parse, we need an intermediate claim.

\textbf{Claim. }Any non-simple infix of a coded word can be uniquely written as $\gamma\delta$ with $\gamma \in L$ and $\delta \in C^*R$. 

\medskip

\emph{Proof of the claim. }
Let $\alpha$ be a non-simple infix of a coded word. 
The existence of the decomposition is immediate as $\alpha$ belongs to $LC^*R$ (by Lemma~\ref{lem:characterization-infix-coded-word}). 
As $\alpha$ is not in $F$ and $F$ contains all words over $\{0,1\}$ that neither contain $00$ nor $11$ as an infix, it follows that $\alpha$ must contain an infix of the form $0^k$ or
$1^k$ for some $k \geq 2$. 
As $w$ is non-simple $k$ can be at most $9$ {and $|\alpha|\geq 11$}. 
By considering the left-most occurrence of maximal size of such an infix  we have that $\alpha$ can uniquely be written as:
\begin{description}
	\item[Case 1] either $\alpha=x 0^k y$ for some $k \in [2,9]$, $x= \{ \varepsilon \} \cup  \{0 ,\varepsilon\} \cdot (1 0)^* 1$ and some $y \in 1 \{ 0,1\}^* \cup \{ \varepsilon \}$.
	\item[Case 2] or $\alpha= x 1^k y$ for some $k \in [2,9]$, $x= \{ \varepsilon \} \cup  \{1 ,\varepsilon\} \cdot (0 1)^* 0$ and some $y \in 0 \{ 0,1\}^* \cup \{ \varepsilon \}$.
\end{description}

We distinguish these two cases.\\

\noindent
\textbf{Assume that we are in Case~1.} 
We have $\alpha=x 0^k y$ for some $x= \{ \varepsilon \} \cup  \{0 ,\varepsilon\} \cdot (1 0)^* 1$ and some $y \in 1 \{ 0,1\}^* \cup \{ \varepsilon \}$. Furthermore assume that $\alpha$ can be written as $\gamma\delta$ for $\gamma \in L$ and $\delta \in C^*R$.

As $\gamma\in L$ and {$\gamma$} is a prefix of $\alpha$ (whose first occurrence of two consecutive letters is $00$), $\gamma$ can only be in
$\{\varepsilon,0,1,00,100,000\} \cup \{1,0,\varepsilon\}(01)^+00$.
For each possible value of $\gamma$, we are going to provide necessary conditions on $x$,$y$ and $k$ for 
$\delta$ to exist. {In fact, we will show that the triple $(x,y,k)$ is uniquely determined by $\gamma$,
hence we will have established the unicity of the decomposition in Case~1.}

\textbf{Case $\gamma=\varepsilon$.} For a suitable $\delta$ to exist,  
$\alpha$ must belong to $C^*R$. In particular, as $|\alpha| \geq 2$, it must start with two occurrences of the same letter. This is only possible if $x=\varepsilon$, so we have $\alpha=\delta=0^k y$. Moreover as $k \leq 9$ and $|\alpha| \geq 11$, it must be the case that $|y|\geq 2$.

We proceed by case distinction on $k \in [2,9]$. The cases $k=2,5,6$ and $9$ can be excluded as $0^2 1$,$0^5 1$, $0^6 1$ and $0^9 1$ are all not the prefix of any word in $C^*R$. 
For $k=3$, $y$ must start\footnote{Recall that $y$ as length at least 2 and starts with 1.} with $10$ as $00011$ is not the prefix of any word in $C^*R$. For $k=4$, $y$ must start with $11$ as
$000010$ is not the prefix of any word in $C^*R$. For $k=7$, $y$ must start with $10$ as $000000011$ is not the prefix of any word in $C^*R$. For $k=8$, $y$ must start with $11$ as $0000000010$ is not the prefix of any word in $C^*R$.

These conditions are summarised in the equation below. In the following cases, we will only provide the summary of the conditions as their proof is similar. 
{Recall that $\sqsubseteq$ denotes the prefix relation on words.}

\[
\gamma=\varepsilon \quad \Rightarrow \quad 
\left\{ \begin{array}{ccccccc}
 & x=\varepsilon & \textrm{and}& k=3 & \textrm{and} & 10 \sqsubseteq y \\ 	
\textrm{or} & x=\varepsilon & \textrm{and} & k=4 & \textrm{and} &11 \sqsubseteq y \\ 	
\textrm{or} & x=\varepsilon & \textrm{and}& k=7 & \textrm{and} &10 \sqsubseteq y \\ 	
\textrm{or} & x=\varepsilon &\textrm{and} & k=8 & \textrm{and}& 11 \sqsubseteq y \\ 	
\end{array}
\right.
\]

\textbf{Case $\gamma=1$.} 
It must the case that $x=1$ and that $k$ belongs to $\{3,4,7,8\}$. Hence
as $|\alpha|\geq 11$, $y$ must have length at least $2$.

\[
\gamma=1 \quad \Rightarrow \quad 
\left\{ \begin{array}{ccccccc}
 & x=1 & \textrm{and}& k=3 & \textrm{and} & 10 \sqsubseteq y \\ 	
\textrm{or} & x=1 & \textrm{and} & k=4 & \textrm{and} &11 \sqsubseteq y \\ 	
\textrm{or} & x=1 & \textrm{and}& k=7 & \textrm{and} &10 \sqsubseteq y \\ 	
\textrm{or} & x=1 & \textrm{and} & k=8 & \textrm{and}& 11 \sqsubseteq y \\ 	
\end{array}
\right.
\]

\textbf{Case $\gamma=0$.} 

\[
\gamma=0 \quad \Rightarrow \quad 
\left\{ \begin{array}{ccccccc}
 & x=\varepsilon & \textrm{and}& k=4 & \textrm{and} & 10 \sqsubseteq y \\ 	
\textrm{or} & x=\varepsilon & \textrm{and} & k=5 & \textrm{and} &11 \sqsubseteq y \\ 	
\textrm{or} & x=\varepsilon & \textrm{and}& k=8 & \textrm{and} &10 \sqsubseteq y \\ 	
\textrm{or} & x=\varepsilon & \textrm{and} & k=9 & \textrm{and}& 11 \sqsubseteq y \\ 	
\end{array}
\right.
\]

\textbf{Case $\gamma=00$.} 

\[
\gamma=00 \quad \Rightarrow \quad 
\left\{ \begin{array}{ccccccc}
& x=\varepsilon & \textrm{and}& k=2 & \textrm{and} & 11 \sqsubseteq y \\
\textrm{or} & x=\varepsilon & \textrm{and}& k=5 & \textrm{and} & 10 \sqsubseteq y \\ 	
\textrm{or} & x=\varepsilon & \textrm{and} & k=6 & \textrm{and} &11 \sqsubseteq y \\ 	
\textrm{or} & x=\varepsilon & \textrm{and}& k=9 & \textrm{and} &10 \sqsubseteq y \\ 	
\end{array}
\right.
\]

\textbf{Case $\gamma=100$.} 

\[
\gamma=100 \quad \Rightarrow \quad 
\left\{ \begin{array}{ccccccc}
& x=1 & \textrm{and}& k=2 & \textrm{and} & 11 \sqsubseteq y \\
\textrm{or} & x=1 & \textrm{and}& k=5 & \textrm{and} & 10 \sqsubseteq y \\ 	
\textrm{or} & x=1 & \textrm{and} & k=6 & \textrm{and} &11 \sqsubseteq y \\ 	
\textrm{or} & x=1 & \textrm{and}& k=9 & \textrm{and} & {1} \sqsubseteq y \\ 	
\end{array}
\right.
\]
\textbf{Case $\gamma=000$.} 

\[
\gamma=000 \quad \Rightarrow \quad 
\left\{ \begin{array}{ccccccc}
& x=\varepsilon & \textrm{and}& k=3 & \textrm{and} & 11 \sqsubseteq y \\
\textrm{or} & x=\varepsilon & \textrm{and}& k=6 & \textrm{and} & 10 \sqsubseteq y \\ 	
\textrm{or} & x=\varepsilon & \textrm{and} & k=7 & \textrm{and} &11 \sqsubseteq y \\ 	
\end{array}
\right.
\]

\textbf{Case $\gamma \in \{1,\varepsilon\}(01)^+00$.} 
\[
{\gamma \in \{1,\varepsilon\}(01)^+00  \quad \Rightarrow \quad 
x00=\gamma\quad\Rightarrow\quad x\in\{1,\varepsilon\}(01)^+\;\; \text{and} \;\; k\in\{2,5,6,9\}}
\]



\textbf{Case $\gamma \in 0 (01)^+ 00$.}
\[
\gamma \in \{1,\varepsilon\}(01)^+00  \quad \Rightarrow \quad 
{x=\varepsilon\quad\textrm{and}\quad k=2 \quad\textrm{and}\quad 10 \sqsubseteq y }
\]

\noindent
It is easy to check that all the cases are mutually exclusive.\\

\noindent
\textbf{Assume that we are in Case~2.} 
We have $\alpha=x 1^k y$ for some $k \in [2,9]$, $x= \{ \varepsilon \} \cup  \{1 ,\varepsilon\} (0 1)^* 0$ and some $y \in 0 \{ 0,1\}^* \cup \{ \varepsilon \}$. 
Furthermore assume that $\alpha$ can be written as $\gamma\delta$ for $\gamma \in L$ and $\delta \in C^*R$.

As $\gamma$ is a prefix of $\alpha$ and belongs to $L$, it can only be in $\{\varepsilon,0,1,11,111\} \cup \{1,\varepsilon\}(01)^+11$.
For each possible value of $\gamma$, we are going to provide necessary conditions on $x$,$y$ and $k$ for $\delta$ to exist. 
{As in Case~1 we will show that the triple $(x,y,k)$ is uniquely determined by $\gamma$,
hence we will have established the unicity of the decomposition in Case~2.}


\textbf{Case $\gamma=\varepsilon$.}
%
%

\[
\gamma=\varepsilon \quad \Rightarrow \quad 
\left\{ \begin{array}{ccccccc}
 & x=\varepsilon & \textrm{and}& k=2 & \textrm{and} & 01 \sqsubseteq y \\ 	
\textrm{or} & x=\varepsilon & \textrm{and} & k=4 & \textrm{and} & 00 \sqsubseteq y \\ 	
\textrm{or} & x=\varepsilon & \textrm{and}& k=6 & \textrm{and} &01 \sqsubseteq y \\ 	
\textrm{or} & x=\varepsilon &\textrm{and} & k=8 & \textrm{and}& 00 \sqsubseteq y \\ 	
\end{array}
\right.
\]

\textbf{Case $\gamma=0$.} 

\[
\gamma=0 \quad \Rightarrow \quad 
\left\{ \begin{array}{ccccccc}
 & x=0 & \textrm{and}& k=2 & \textrm{and} & 01 \sqsubseteq y \\ 	
\textrm{or} & x=0 & \textrm{and} & k=4 & \textrm{and} & 00 \sqsubseteq y \\ 	
\textrm{or} & x=0 & \textrm{and}& k=6 & \textrm{and} &01 \sqsubseteq y \\ 	
\textrm{or} & x=0 &\textrm{and} & k=8 & \textrm{and}& 00 \sqsubseteq y \\ 	
\end{array}
\right.
\]

\textbf{Case $\gamma=1$.} 

\[
\gamma=1 \quad \Rightarrow \quad 
\left\{ \begin{array}{ccccccc}
 & x=\varepsilon & \textrm{and}& k=3 & \textrm{and} & 01 \sqsubseteq y \\ 	
\textrm{or} & x=\varepsilon & \textrm{and} & k=5 & \textrm{and} & 00 \sqsubseteq y \\ 	
\textrm{or} & x=\varepsilon & \textrm{and}& k=7 & \textrm{and} &01 \sqsubseteq y \\ 	
\textrm{or} & x=\varepsilon &\textrm{and} & k=9 & \textrm{and}& 00 \sqsubseteq y \\ 	
\end{array}
\right.
\]

\textbf{Case $\gamma=11$.} 

\[
\gamma=11 \quad \Rightarrow \quad 
\left\{ \begin{array}{ccccccc}
 & x=\varepsilon & \textrm{and}& k=2 & \textrm{and} & 00 \sqsubseteq y \\ 	
\textrm{or} & x=\varepsilon & \textrm{and}& k=4 & \textrm{and} & 01 \sqsubseteq y \\ 	
\textrm{or} & x=\varepsilon & \textrm{and} & k=6 & \textrm{and} & 00 \sqsubseteq y \\ 	
\textrm{or} & x=\varepsilon & \textrm{and}& k=8 & \textrm{and} &01 \sqsubseteq y \\ 	
\end{array}
\right.
\]

\textbf{Case $\gamma=111$.} 

\[
\gamma=111 \quad \Rightarrow \quad 
\left\{ \begin{array}{ccccccc}
& x=\varepsilon & \textrm{and}& k=3 & \textrm{and} & 00 \sqsubseteq y \\ 
\textrm{or} & x=\varepsilon & \textrm{and}& k=5 & \textrm{and} & 01 \sqsubseteq y \\ 	
\textrm{or} & x=\varepsilon & \textrm{and} & k=7 & \textrm{and} & 00 \sqsubseteq y \\ 	
\textrm{or} & x=\varepsilon & \textrm{and}& k=9 & \textrm{and} &01 \sqsubseteq y \\ 	
\end{array}
\right.
\]

\textbf{Case $\gamma \in \{1,\varepsilon\}(01)^+11$.} 
\[
\gamma \in \{1,\varepsilon\}(01)^+11  \;\Rightarrow \;
{x111=\gamma},\;x\not=\varepsilon\;\textrm{and}\;
\left\{ \begin{array}{ccccccc}
 &  & & k=3 & \textrm{and} & 00 \sqsubseteq y  \quad\textrm{if $|y|\geq 2$}\\
 \textrm{or}&  & & k=5 & \textrm{and} & 01 \sqsubseteq y  \quad\textrm{if $|y|\geq 2$}\\ 	
\textrm{or} & &  & k=7 & \textrm{and} & 00 \sqsubseteq y \quad\textrm{if $|y|\geq 2$}\\ 	
\textrm{or} &  & & k=9 & \textrm{and} &01 \sqsubseteq y \quad \textrm{if $|y|\geq 2$}\\ 		
\end{array}
\right.
\]


\noindent
It is easy to check that all the cases are mutually exclusive. 

\emph{End of the proof of the claim.}

\noindent
We can now prove that parses are unique. Let $p_1=(\ell_1,u_1,r_1)$ and $p_2=(\ell_2,u_2,r_2)$ be two parses of $\alpha$. 
As $(\ell_1,\psi(u_1)r_1)$ and $(\ell_2,\psi(u_2) r_2)$ are both decompositions of $\alpha$ in $L\times  C^*R$, we have $\ell_1=\ell_2$ by the previous claim  and thus $\psi(u_1)r_1=\psi(u_2)r_2$.
It suffices to show $u_1=u_2$ since this implies $r_1=r_2$. 
Towards a contradiction, assume $u_1\not= u_2$.
Let $x$ be the longest common prefix of $u_1$ and $u_2$. Thus, $u_1=xv_1$ and $u_2=xv_2$ for some words $v_1,v_2$ of which at least one is non-empty. 
Without loss of generality assume $v_1\not=\varepsilon$, in particular $v_1=av_1'$ for some letter $a$.
In case $v_2\not=\varepsilon$, then $v_2=bv_2'$ for some letter $b$ with $a\not=b$ and hence either $\psi(a)\sqsubseteq\psi(b)$ or $\psi(b)\sqsubseteq\psi(a)$, 
thus contradicting that $\psi$ is an infix code.
In case $v_2=\varepsilon$ it follows $r_2=\psi(v_1)r_1=\psi(av_1')r_1$, 
which implies that $\psi(a)$ is prefix of $r_2\in R$, again
contradicting that $\psi$ is an infix code.
\fi
\end{proof}


Thus, we will refer to the unique parse of a non-simple infix $\alpha$ 
of a coded word as {\em the parse of $\alpha$}.

The next lemma states that occurrences of codings of symbols of order
strictly larger than one in a coded word can be related with occurrences of this 
symbol in the word that has been coded.

\begin{lemma}
\label{lemma:occurences-letter-parse}
Let $\alpha$ be a non-simple infix {of }some coded word {and let $(\ell,u,r)$ be its parse}.

If $\alpha$ contains $n>1$ occurrences of $\psi(x)$ for some $x \in \Sigma \setminus \Sigma_1$ then $u$ contains $n$ occurrences of $x$.
\end{lemma}
\begin{proof}
Let $\alpha$ a non-simple infix of some coded word and let $p=(\ell,u,r)$ be its parse. Let $x$ be a letter in $\Sigma \setminus \Sigma_1$ such that 
$\alpha$ contains $n>1$ occurrences of $\psi(x)$.

By definition of the parse $p$, we have $\alpha = \ell \psi(u) r$.
If we write $u=u_1 \cdots u_m$ with $m \geq 0$ and $u_i \in \Sigma$ for all $i \in [1,m]$, we can write $\alpha$ as follows,
\[
 \alpha =\alpha_0 \alpha_1 \cdots \alpha_m \alpha_{m+1},
\] 
where
\begin{itemize}
	\item $\alpha_0=\ell$,
	\item $\alpha_i= \psi(u_i)$ for all $i \in [m]$,
	\item and $\alpha_{m+1}=r$.
\end{itemize}

Let $m_1<\cdots<m_n$ be an enumeration of the $n$ occurrences of $\psi(x)$ in
$\alpha$. 
For all $i \in {[n]}$, we denote by $q_i$ the maximal integer satisfying
$m_i \geq \sum_{j=0}^{q_i-1} |\alpha_j|$.  

\textbf{Claim. }For all $i \in [n]$ we have $0<q_i\leq m$ and $m_i=\sum_{j=0}^{q_i-1} |\alpha_j|$.

\noindent 
\emph{Proof of the claim.}  Let $i \in [1,n]$.

{Let us first show that $q_i\not=0$.}
Assume towards a contradiction that $q_i=0$. By maximality of $q_i$, $\alpha_0=\ell$ cannot be empty. This implies that $\psi(x)$ can be written as $\ell' \alpha_1 \cdots \alpha_k r'$ with $\ell'$ a non-empty suffix of $\ell$, $k\geq 0$ and $r'$ a prefix of $\alpha_{k+1}$. As $C$ is an infix code, $k$ is necessarily equal to $0$. Hence 
$\psi(x)=\ell'r'$. In particular $\psi(x) \in LR \cap C$. 
{In Lemma~\ref{lem:observations}, we remarked that} $LR \cap C = \{0000,1111\}${,} which brings a contradiction with the fact that 
$x\in\Sigma\setminus\Sigma_1$.

{Let us next show $q_i\not=m+1$.}
Assume towards a contradiction that $q_i=m+1$. In this case, $r \in R$ would contain $\psi(x)$ as an infix which contradicts the fact that $C$ is an infix code.

{Let us finally show $m_i\leq \sum_{j=0}^{q_i-1} |\alpha_j|$ (and thus $m_i= \sum_{j=0}^{q_i-1} |\alpha_j|$).}
Assume towards a contradiction that $m_i >\sum_{j=0}^{q_i-1} |\alpha_j|$. By definition of $q_i$, this implies that $\psi(x)$ can be written as $\ell' \alpha_{q_i+1} \cdots \alpha_{q_i+k} r'$ with $\ell'$ a non-empty suffix of $\alpha_{q_i}$, $k\geq 0$ and $r'$ a prefix of $\alpha_{q_{i}+k+1}$. As $C$ is an infix code, $k$ is necessarily equal to $0$. Hence 
$\psi(x)=\ell'r'$. In particular $\psi(x) \in LR \cap C$. 
{As above,  we have from Lemma~\ref{lem:observations}, that} $LR \cap C = \{0000,1111\}${,} which brings a contradiction with the fact that $x$ is not of order $1$.

\noindent
\emph{End of the proof of the claim.}

Using the claim, it follows that either $\psi(x)$ is a prefix of $\psi(u_{q_i})$ 
or conversely that $\psi(u_{q_i})$ is a prefix of $\psi(x)$. 
{As} $C$ is an infix code, this is only possible if $\psi(x)=\psi(u_{q_i})$ 
and hence $u_{q_i}=x$.

Again using the claim, we have that $q_1<\cdots<q_n$ (as $m_1<\cdots<m_n$). 
Hence we have shown that $u$ contains at least $n$ occurrences of $x$. 
Clearly, $u$ cannot contain more than $n$ occurrences of $x$ as each occurrence of $x$ in $u$ induces an occurrence of $\psi(x)$ in $\alpha$.
\end{proof}

Let $w=w_0 \cdots w_{|w|-1} \in \Sigma^*$ and $p=(\ell,u=u_0\cdots u_{|u|-1} ,r)$ 
be a parse, an {\em occurence of $p$ in $w$} is an occurrence $m$ of $u$ in $w$ such that whenever $\ell$ is non-empty we have $m \neq 0$ and 
$\ell$ is {a} suffix of $\psi(w_{m-1})$ and similarly whenever $r$ is non-empty
we have $m + |u| < |w|$ and $r$ is a prefix of $\psi(w_{m+|u|})$.

{
\begin{remark}
	In the previous lemma, the requirement that the order of the symbol is strictly greater than 1 is necessary. For instance consider the coded word $w=\psi(0_20_3)=00010000010100$. If we take $\alpha$ to be $w$ which is non-simple, $\alpha$ contains $\psi(0_1)=0000$ as an infix but $0_1$ does not occur in its unique parse of $(\varepsilon,0_2 0_3,\varepsilon)$. 
\end{remark}
}

{The next lemma shows that for} a word $w \in \Sigma^*$ and a non-simple infix $\alpha$ of $\psi(w)$, there is 
a one-to-one correspondence between the occurrences of $\alpha$ in $\psi(w)$ and the occurrences of its parse $p_\alpha$ in $\psi(w)$.

\begin{lemma}
\label{lem:bijection_occurrences}
For any word $w \in \Sigma^*$ and any non-simple infix $\alpha$ of $\psi(w)$, there is a unique order-preserving bijection between the occurrences of $\alpha$ in $\psi(w)$ and the occurrences of its parse $p$ in $w$.	
\end{lemma}

\begin{proof}
Let $w=w_0w_1 \cdots w_{n-1}$, $n\geq1$ a non-empty word over $\Sigma$ and let $\alpha$ be a non-simple infix of the word $\psi(w)$. Consider the unique parse $p=(\ell,u,r)$ of the infix $\alpha$. To each occurrence $m$ of the parse $p$ in $w$, we associate
the occurrence $\rho(m)=(\sum_{i=0}^{m-1} |\psi(w_i)|) - |\ell|$ of the word $\alpha$ in $\psi(w)$. The mapping $\rho$, from the set of occurrences of $p$ in $w$ to the set of occurrences of $\alpha$ in $\psi(w)$, is order-preserving and injective. It remains to show that it is surjective.

Let $h$ be an occurrence of $\alpha$ in $\psi(w)$. By definition, there exist 
two words $x,y \in \{0,1\}^*$ such that $\psi(w)=x \alpha y$ and $|x|=h$. Consider the greatest integer $m_0 \in [0,n-1]$ such that
\[
h \geq \sum_{i=0}^{m_0-1} |\psi(w_i)|{.}
\]

We will show that $m_0$ is an occurrence of the parse $p$ in $w$ and hence 
that $\tau(m_0)=h$. 
Remarking that $\psi(w)$ is equal to both 
$x \alpha y$ and $\psi(w_0) \cdots \psi(w_{n-1})$, there must exist
$k \geq 0$ such that $\alpha = \ell' \psi(w_{m_0}) \cdots \psi(w_{m_0+k-1}) r'$,
where
\begin{itemize}
\item $\ell'$ is empty if $h= \sum_{i=0}^{m_0-1} |\psi(w_i)|$ and 
 is the suffix of length $h - \sum_{i=0}^{m_0-1} |\psi(w_i)|$ of $\psi(w_{m_0-1})$
 otherwise, 
 \item and $r'$ is empty if $m_0+k-1=n-1$ and a prefix of 
 $\psi(w_{m_0+k})$ otherwise. 
\end{itemize}
It follows that $(\ell',w_{m_0} w_{m_0+1} \cdots w_{m_0+k-1},r')$ is a parse of $\alpha$ occurring at $m_0$ in $w$. 
The lemma now follows from the unicity of the parse.
\end{proof}

\begin{definition}
	For an \modified{occurrence} $m$ of a parse $p=(\ell,u,r)$ in $w$, we define its 
{\emph{context}} $[p]_m$ as the word in $\Sigma^*$ equal to $w[m-\delta_0,m+|u|+\delta_1]$ where $\delta_0=0$ if $\ell=\varepsilon$ and $\delta_0=1$ otherwise and $\delta_1=0$ if $r=\varepsilon$ and $\delta_1=1$ otherwise.
\end{definition}

By definition{, } the context $c$ of some occurrence of a parse $p=(\ell,u,r)$ in $w$ is an infix of $w$,
{that itself contains $u$ as an infix}. 
Moreover{, }the value $\alpha$ of $p$ is an infix of $\psi(c)$.

\input{proof-binary.tex}

%% file: proof-binary.tex
\subsection{Upper bound on the Zimin index}
\label{ssec:proof-non-binary}
{We are now ready to upper-bound the Zimin index of the code of higher-order counters. 
Due to the nature of our coding $\psi$ we need to prove a slightly stronger inductive statement
that takes into the account the code of a symbol of order $n+1$ directly before or directly after
the code of a counter of order $n$.}
\begin{theorem}
For all $n \geq 2$ and for all $i \in [0,\max{n}-1]$, 
\[
\begin{array}{lclclcl}
\Zimin{\cntb{i}{n} \psi(0_{n+1})} & \leq & n+1, &  \quad & \Zimin{\cntb{i}{n} \psi(1_{n+1})} & \leq & n+1, \\
\Zimin{\psi(0_{n+1}) \cntb{i}{n}}  & \leq & n+1, & & \Zimin{\psi(1_{n+1}) \cntb{i}{n}}  & \leq & n+1. \\ 
\end{array}
\]
\end{theorem}

\begin{proof}
We proceed by induction on $n$. For the cases $n=2$ and $n=3$, the property is checked using a computer program. 
	Remark that the reason we start the induction at $3$ is to be able to apply the upper bound from Lemma~\ref{lem:Zimin-simple}. 


For the induction step assume that the property holds for some $n \geq 3$ and let us show that it holds for $n+1$. Let $i \in [0,\max{n+1}-1]$, we have to show that 
\[
\begin{array}{lclclcl}
\Zimin{\cntb{i}{n+1} \psi(0_{n+2})} & \leq & n+2, & & 
\Zimin{\cntb{i}{n+1} \psi(1_{n+2})} & \leq & n+2, \\
 \Zimin{\psi(0_{n+2})\cntb{i}{n+1}}  & \leq & n+2, & &  
 \Zimin{\psi(1_{n+2})\cntb{i}{n+1}}  & \leq & n+2.\\ 
 \end{array}
\] 
 
 We start by showing that $\Zimin{\cntb{i}{n+1}} \leq n+2$. Let $\alpha \beta \alpha$ be an infix of $\cntb{i}{n+1}$ for some non-empty $\alpha$ and $\beta$. It is enough to show that $\ZiminType{\alpha} \leq n+1$. By Lemma~\ref{lem:Zimin-simple}, we only need to consider the case when $\alpha$ is non-simple.
	Let ${p}=(\ell,u,r)$ be the parse of $\alpha$\modified{, whose uniqueness} is guaranteed by 
	Lemma~\ref{lem:unique-parse}.

Let $m$ be an occurrence of $\alpha\beta\alpha$ in $\cntb{i}{n+1}$.
In particular, $m$ and $m+|\alpha\beta|$ are two occurrences of $\alpha$ in  $\cntb{i}{n+1}$. By Lemma~\ref{lem:bijection_occurrences}, there are two corresponding occurrences $m_1$ and $m_2$ of the parse $p$ in $\cnt{i}{n+1}$. Consider the contexts $c_1$  and $c_2$ of $p$ that correspond to the occurrences $m_1$ and $m_2$,
respectively.
Note that without further hypothesis $c_1$ and $c_2$ are not necessarily equal.

 We distinguish cases depending on the number of occurrences of a symbol of order $n+1$ in $c_1$.
 \medskip
 \noindent

\textbf{If $c_1$ does not contain any symbol of order $n+1$.} As $c_1$ is an infix of $\cnt{i}{n+1}$ 
{and since by assumption $c_1$} does not contain any symbol of order $n+1$, it must be an 
infix of some $\cnt{j}{n}$ with $j \in [0,\max{n}-1]$.
By definition of the context of a parse, $\alpha$ is an infix of $\psi(c_1)$ and hence of $\cntb{j}{n}$. 
Thus, 
\[
\ZiminType{\alpha} \leq \Zimin{\cntb{j}{n}} \leq \Zimin{\cntb{j}{n}\psi(0_{n+1}}
\leq n+1,
\]
where the last inequality follows from induction hypothesis.
 \medskip
 \noindent

\textbf{If $c_1$ contains at least two symbols of order $n+1$.} We will show that this situation cannot occur. By definition of $\cnt{i}{n+1}$, $c_1$ contains an infix of the form $b\cnt{j}{n}b'$ for some $j \in [0,\max{n}-1]$ and $b,b' \in \{0_{n+1},1_{n+1}\}$. The center $u$ of the parse $p=(\ell,u,r)$ must therefore contain $\cnt{j}{n}$ as an infix. As there are two occurrences of $u$ in $\cnt{i}{n+1}$\footnote{Recall that there are two occurrences of $p$ in $\cnt{i}{n+1}${.}}, this would imply that $\cnt{j}{n}$ has two occurrences in $\cnt{i}{n+1}${,} which brings the contradiction (using Lemma~\ref{lemma:basics-counters}).
 \medskip
 \noindent

\textbf{If $c_1$ contains one and only one symbol of order $n+1$.} 

As $c_1$ is an infix of $\cnt{i}{n+1}$ with one order $n+1$ symbol, there exists 
$k_0 \in [0,\max{n}-2]$ and some $b\in\{0_{n+1},1_{n+1}\}$ such that
\[
c_1 = x b y, \text{ where $x\in\Sigma_{n}^*$ is a suffix of $\cnt{k_0}{n}$
and 
$y\in\Sigma_n^*$ is a prefix of $\cnt{k_0+1}{n}$.}
\]

Remark that if $x$ or $y$ are empty, we can conclude using induction hypothesis. 
Indeed in {these} cases, $c_1$ is an infix
of either $\cnt{{k_0}}{n} b$ or 
$b \cnt{{k_0+1}}{n}$. Hence $\alpha$, which is an infix of $\psi(c_1)$, 
is also an infix of either  $\cntb{{k_0}}{n} \psi(b)$ or  $\psi(b) \cntb{{k_0+1}}{n}$. As by induction hypothesis both have Zimin index at most $n+1$, we can conclude using Lemma~\ref{lemma:zimin-monotone} that $\ZiminType{\alpha}\leq n+1$.

From now on, we assume that {both} $x$ and $y$ are non-empty. In particular, the center $u$ of the 
parse $p=(\ell,u,r)$ contains $b$ and can therefore be uniquely written as ${u=}\underline{x}b\underline{y}$.
In summary, we have
\[
\begin{array}{rcl}
c_1    & = & xby, \\
\alpha & = & \ell \psi(\underline{x}) \psi(b) \psi(\underline{y}) r, \\
x      & = & s \underline{x}, \\
y      & = & \underline{y}t. \\\	
\end{array}
\]
for some $s$ and $t$ such that:
\begin{itemize}
\item $s=\varepsilon$ if $\ell=\varepsilon$ and otherwise $s \in \Sigma$ with $\ell$ is a suffix of $\psi(s)$.
\item $t=\varepsilon$ if $r=\varepsilon$ and otherwise $t\in \Sigma$ with $r$ is a prefix of $\psi(t)$.
\end{itemize}

\bigskip
\textbf{Claim~1.} The context $c_2$ (of the second occurence of $\alpha$) is equal to $c_1$. 

\medskip

\noindent
\emph{Proof of the Claim 1:} Similarly as for $c_1$, the context $c_2$ can be written as $s'\underline{x}b\underline{y}t'$ for some $s'$ and $t'$ such that:
\begin{itemize}
\item $s'=\varepsilon$ if $\ell=\varepsilon$ and otherwise $s' \in \Sigma$ with $\ell$ is a suffix of $\psi(s')$.
\item $t'=\varepsilon$ if $r=\varepsilon$ and otherwise $t' \in \Sigma$ with $t$ is a prefix of $\psi(t')$.
\end{itemize}
Towards a contradiction, assume $c_1$ and $c_2$ are different. It is either the case that $s\neq s'$ or the $t\neq t'$.
As both cases can be shown analogously, we only consider the first one and assume 
that $s\neq s'$. In particular, {without loss of generality we may assume that $\ell$ is non-empty}.

The symbols $s$ and $s'$ occur in $\cnt{i}{n+1}$ at the same distance of an order $n+1$ symbol and by Lemma~\ref{lem:order-from-position} must have the same order. Furthermore the last symbol of their encoding by $\psi$ is the same (it is the last symbol of $\ell$). By the definition of $\psi$,
$s$ and $s'$ are either both from {$\{0_k\mid k\geq 1\}$} or both 
from {$\{1_k\mid k\geq 1\}$}. 
This proves that $s$ and $s'$ are equal which brings the contradiction.

\noindent
\emph{End of the proof of Claim 1.}

\bigskip
\noindent
As $c_1=c_2=xby$ and as $b$ belongs to the center $u$ of the parse, the infix $\alpha$ can be written as 
\[
\alpha=\tilde{x}\psi(b)\tilde{y} 
\]
where $\tilde{x}$ is a suffix of $\psi(x)$ and $\tilde{y}$ is a prefix of $\psi(y)$.
\bigskip

\noindent
\textbf{Claim 2.} There exists $j_0 \in [0,\max{n}-1]$ 
and a non-empty $\chi$ such that $\tilde{y}\chi\tilde{x} = \cntb{j_0}{n}$.

\medskip

\noindent
\emph{Proof of Claim 2.}
We proceed along the same lines as in the proof of 
Theorem~\ref{thm:zimin-non-binary-counters}. 
Consider the morphism $\varphi$ that erases all symbols 
in $\Sigma_{n-1}$ and replaces $0_n$ 
and $1_n$ by $0$ and $1$ respectively. 
That is, we can write $\varphi(x)$ and $\varphi(y)$ as follows,
\[
\begin{array}{lcl}
	\varphi(x) & = & b_{\max{n-1}-\ell_0} \cdots b_{\max{n-1}-1} \\
	\varphi(y) &=  & c_0 \cdots c_{\ell_1-1} \\
\end{array}
\]
where $b_{\max{n-1}-k}\in\{0,1\}$ for all
$k\in[1,\ell_0]$ and $c_k\in\{0,1\}$ for all $k\in[0,\ell_1-1]$.

With the same proof as in  Theorem~\ref{thm:zimin-non-binary-counters}, we 
show that
\begin{equation}
	\ell_0 + \ell_1 < \max{n-1}.
\end{equation} 

By the same reasoning as in the proof of 
Theorem~\ref{thm:zimin-non-binary-counters}, there exists 
$j_0 \in [0,\max{n}-1]$ and non-empty $\xi$ such that $\cnt{j_0}{n}=y\xi x$. By applying $\psi$ and recalling that $\tilde{x}$ is a suffix of $\psi(x)$ and $\tilde{y}$ is a prefix of $\psi(y)$, we can conclude.

\emph{End of the proof of Claim 2.}

\bigskip

Let us now consider an arbitrary decomposition of $\alpha$ as 
$\delta\gamma\delta$ for non-empty $\delta$ and $\gamma$. 
Recall that it is enough to show that $\ZiminType{\alpha} \leq n+1$ or
that $\ZiminType{\delta} \leq n$.

There are several cases to consider depending on how the two decompositions $\tilde{x}\psi(b) \tilde{y}$ and $\delta \gamma \delta$ overlap.
\medskip

\textbf{Case 1: $|\tilde{x} \psi(b)| \leq |\delta|$.} 

\medskip

\begin{tikzpicture}[scale=0.55]
\node at (-2,1) {$\alpha=$};
\draw (0,0) rectangle (6,1);
\node at (3,0.5) {$\delta$}; 
\draw (6,0) rectangle (12,1);
\node at (9,0.5) {$\gamma$}; 
\draw (12,0) rectangle (18,1);
\node at (15,0.5) {$\delta$}; 

\draw (0,1) rectangle (2,2);
\node at (1,1.5) {$\tilde{x}$};
\draw (2,1) rectangle (4,2);
\node at (3,1.5) {$\psi(b)$};
\draw (4,1) rectangle (18,2);
\node at (11,1.5) {$\tilde{y}$};

\end{tikzpicture}

This situation cannot occur under our hypothesis.
Indeed, $\alpha$ would contain two occurences of 
$\psi(b)$ which by Lemma~\ref{lemma:occurences-letter-parse} implies that the center of its parse contains two occurences of the order $n+1$ symbol $b$. 
This brings a contradiction with the fact that the context 
$c_1$ contains exactly one symbol of order $n+1$.
\medskip

\textbf{Case 2: $|\tilde{x}|\leq |\delta|$  and $|\delta|<|\tilde{x}\psi(b)|\leq |\delta\gamma|$.} 

\medskip

\begin{tikzpicture}[scale=0.55]
\node at (-2,1) {$\alpha=$};
\draw (0,0) rectangle (6,1);
\node at (3,0.5) {$\delta$}; 
\draw (6,0) rectangle (12,1);
\node at (9,0.5) {$\gamma$}; 
\draw (12,0) rectangle (18,1);
\node at (15,0.5) {$\delta$}; 

\draw (0,1) rectangle (2,2);
\node at (1,1.5) {$\tilde{x}$};
\draw (2,1) rectangle (8,2);
\node at (5,1.5) {$\psi(b)$};
\draw (8,1) rectangle (18,2);
\node at (13,1.5) {$\tilde{y}$};

\draw[decorate,decoration={brace,amplitude=10pt},yshift=4pt,xshift=0pt]
(2,2) -- (6,2) node [black,midway,yshift=0.6cm] {
$z_1$};
\draw[decorate,decoration={brace,amplitude=10pt},yshift=4pt,xshift=0pt]
(6,2) -- (8,2) node [black,midway,yshift=0.6cm] {
$z_2$};

\draw[decorate,decoration={brace,amplitude=10pt},yshift=-4pt,xshift=0pt]
(8,0) -- (6,0) node [black,midway,yshift=-0.6cm] {
$\gamma_1$};
\draw[decorate,decoration={brace,amplitude=10pt},yshift=-4pt,xshift=0pt]
(12,0) -- (8,0) node [black,midway,yshift=-0.6cm] {
$\gamma_2$};

\end{tikzpicture}

In this case, $\psi(b)$ can be written as $z_1 z_2$ such that $\delta=\tilde{x} z_1$ and $\gamma$ as $\gamma_1 \gamma_2$ such that $\tilde{y}=\gamma_2 \delta$
and $z_1\not=\varepsilon$.

By Claim 2, there exists $j_0 \in [0,\max{n}-1]$ and a non-empty $\chi$ such that $\tilde{y}\chi\tilde{x}=\cntb{j_0}{n}$.

By induction hypothesis, $\cntb{j_0}{n}\psi(b)$ has Zimin index at most $n+1$. 
In particular, 
$\tilde{y}\chi\tilde{x}z_1$, which is a prefix of $\cntb{j_0}{n}\psi(b)$,
also has Zimin index at most $n+1$.

As $\tilde{y}\chi\tilde{x}z_1$ is equal to
$\gamma_2 \delta \chi \delta$
we have that $\delta \chi \delta$ is an infix of a word (\ie $\cntb{j_0}{n}z_1$) of Zimin index at most $n+1$. 
By Fact~\ref{fact:ZiminType-inductive}, this implies that $\ZiminType{\delta} \leq n$ which concludes the case.
\medskip

\textbf{Case 3: $|\tilde{x} |\leq |\delta|$ and $|\tilde{x}\psi(b)|>|\delta\gamma|$.} 

\medskip

\begin{tikzpicture}[scale=0.55]
\node at (-2,1) {$\alpha=$};
\draw (0,0) rectangle (6,1);
\node at (3,0.5) {$\delta$}; 
\draw (6,0) rectangle (12,1);
\node at (9,0.5) {$\gamma$}; 
\draw (12,0) rectangle (18,1);
\node at (15,0.5) {$\delta$}; 

\draw (0,1) rectangle (2,2);
\node at (1,1.5) {$\tilde{x}$};
\draw (2,1) rectangle (16,2);
\node at (9,1.5) {$\psi(b)$};
\draw (16,1) rectangle (18,2);
\node at (17,1.5) {$\tilde{y}$};

\draw[decorate,decoration={brace,amplitude=10pt},yshift=4pt,xshift=0pt]
(2,2) -- (6,2) node [black,midway,yshift=0.6cm] {
$z_1=r$};
\draw[decorate,decoration={brace,amplitude=10pt},yshift=4pt,xshift=0pt]
(6,2) -- (12,2) node [black,midway,yshift=0.6cm] {
$\gamma$};
\draw[decorate,decoration={brace,amplitude=10pt},yshift=4pt,xshift=0pt]
(12,2) -- (16,2) node [black,midway,yshift=0.6cm] {
$z_2=\ell$};

\end{tikzpicture}

In this case $\psi(b)$ can be written as $\psi(b)=z_1 \gamma z_2$ with $z_2$ non-empty such that:
\begin{itemize}
	\item $\delta = \tilde{x}z_1$,
	\item $\delta = z_2\tilde{y}$.
\end{itemize}

First recall that $\alpha=\ell\psi(\underline{x})\psi(b)\psi(\underline{y})r$.
Next, recall that $\tilde{x}=\ell \psi(\underline{x})$, hence 
$\delta = \ell \psi(\underline{x}) z_1$. It follows that 
$(\ell,\underline{x},z_1)$ is a parse of $\delta$.

Finally, recall that $\tilde{y}= \psi(\underline{y}) r$, hence $\delta = z_2 \psi(\underline{y}) r$. It follows that $(z_2,\underline{y},r)$ is a parse of $\delta$.

By the unicity of the parse (Lemma~\ref{lem:unique-parse}), we have $(z_2,\underline{y},r)=(\ell,\underline{x},{z_1})$ and hence $z_2=\ell$, $z_1=r$ and $\underline{x}=\underline{y}$.

We will now show that $\underline{x}=\underline{y}$ is empty. 

Towards a contradiction, assume that $\underline{x}$ is not empty. 
We recall that 
$c_1 = x b y$, where $x\in\Sigma_{n}^*$ is a suffix of $\cnt{k_0}{n}$
and $y\in\Sigma_n^*$ is a prefix of $\cnt{k_0+1}{n}$.

Since $x=s\underline{x}$ it follows that $\underline{x}$ is a suffix of 
$\cnt{k_0}{n}$. 
By definition of $\cnt{k_0}{n}$ we have that
$\underline{x}$ ends with an order $n$ symbol. 
But $\underline{x}=\underline{y}$ is a also prefix of $\cnt{k_0+1}{n}$ 
(which contains an order $n$ symbol) and hence 
starts with $\cnt{0}{n-1}$. 
By Lemma~\ref{lemma:basics-counters}, a suffix of $\cnt{k_0}{n}$ 
starting with $\cnt{0}{n-1}$ is equal to $\cnt{{k_0}}{n}$. Hence $\tilde{x}=\cnt{k_0}{n}$ which is not a prefix of $\cnt{k_0+1}{n}$, which brings the contradiction.

Hence we have $\delta=\ell r=z_2z_1$ and in particular
$|\delta| < |\psi(b)| = 4 + 2n$.
By Lemma~\ref{lemma:zimin-monotone}, we can bound the Zimin index of $\delta$ by
\[
\Zimin{\delta} \leq \lfloor \log_2(2n+4) \rfloor.
\]
As for all $n \geq 3$, $\lfloor\log_2(2n+4)\rfloor\leq n$, we have shown that $\ZiminType{\delta}\leq \Zimin{\delta}\leq n$, which concludes this case.

\textbf{Case 4: $ |\delta| <|\tilde{x} |\leq |\delta\gamma|$ and $|\tilde{x}\psi(b)|\leq |\delta\gamma|$.} 

\medskip

\begin{tikzpicture}[scale=0.55]
\node at (-2,1) {$\alpha=$};
\draw (0,0) rectangle (6,1);
\node at (3,0.5) {$\delta$}; 
\draw (6,0) rectangle (12,1);
\node at (9,0.5) {$\gamma$}; 
\draw (12,0) rectangle (18,1);
\node at (15,0.5) {$\delta$}; 

\draw (0,1) rectangle (8,2);
\node at (4,1.5) {$\tilde{x}$};
\draw (8,1) rectangle (10,2);
\node at (9,1.5) {$\psi(b)$};
\draw (10,1) rectangle (18,2);
\node at (14,1.5) {$\tilde{y}$};

\draw[decorate,decoration={brace,amplitude=10pt},yshift=4pt,xshift=0pt]
(6,2) -- (8,2) node [black,midway,yshift=0.6cm] {
$z_1$};
\draw[decorate,decoration={brace,amplitude=10pt},yshift=4pt,xshift=0pt]
(10,2) -- (12,2) node [black,midway,yshift=0.6cm] {
$z_2$};

\end{tikzpicture}

In this case $\gamma=z_1 \psi(b) z_2$ with $z_1 \neq \varepsilon$ such that $\tilde{x}=\delta z_1$ and $\tilde{y}=z_2\delta$.

By Claim~2, there exists $j_0 \in [0,\max{n}-1]$ 
and a non-empty $\chi$ such that $\tilde{y}\chi\tilde{x} = \cntb{j_0}{n}$.  
We have $\Zimin{{\cntb{j_0}{n}}}\leq\Zimin{{\cntb{j_0}{n}}\psi(0_{n+1})}\leq n+1$,
where the last inequality follows from induction hypothesis.
Hence,
\[
\cntb{j_0}{n}=\tilde{y}\chi\tilde{x} = z_2 \delta \chi \delta z_1
\]
has Zimin index at most $n+1$. This implies that $\delta$ has Zimin type of at most $n$ which 
{concludes} this case.
\medskip

\textbf{Case 5: $ |\delta| <|\tilde{x} |\leq |\delta\gamma|$ and $|\tilde{x}\psi(b)| > |\delta\gamma|$.} 

\medskip

\begin{tikzpicture}[scale=0.55]
\node at (-2,1) {$\alpha=$};
\draw (0,0) rectangle (6,1);
\node at (3,0.5) {$\delta$}; 
\draw (6,0) rectangle (12,1);
\node at (9,0.5) {$\gamma$}; 
\draw (12,0) rectangle (18,1);
\node at (15,0.5) {$\delta$}; 

\draw (0,1) rectangle (8,2);
\node at (4,1.5) {$\tilde{x}$};
\draw (8,1) rectangle (14,2);
\node at (11,1.5) {$\psi(b)$};
\draw (14,1) rectangle (18,2);
\node at (16,1.5) {$\tilde{y}$};

\draw[decorate,decoration={brace,amplitude=10pt},yshift=4pt,xshift=0pt]
(8,2) -- (12,2) node [black,midway,yshift=0.6cm] {
$z_1$};
\draw[decorate,decoration={brace,amplitude=10pt},yshift=4pt,xshift=0pt]
(12,2) -- (14,2) node [black,midway,yshift=0.6cm] {
$z_2$};

\draw[decorate,decoration={brace,amplitude=10pt},yshift=-4pt,xshift=0pt]
(8,0) -- (6,0) node [black,midway,yshift=-0.6cm] {
$\gamma_1$};
\draw[decorate,decoration={brace,amplitude=10pt},yshift=-4pt,xshift=0pt]
(12,0) -- (8,0) node [black,midway,yshift=-0.6cm] {
$\gamma_2$};

\end{tikzpicture}

This case is the symmetric {to} Case~2. 
One can write $\psi(b)$ as 
$z_1 z_2$ such that $\delta=z_2 \tilde{y}$ and 
$\gamma$ as $\gamma_1 \gamma_2$ such that $\tilde{x}=\delta \gamma_1$,
where $\gamma_1\not=\varepsilon$ and $z_2\not=\varepsilon$.

By Claim 2, there exists $j_0 \in [0,\max{n}-1]$ and a non-empty $\chi$ such that $\tilde{y}\chi\tilde{x}=\cntb{j_0}{n}$.

By induction hypothesis, $\psi(b)\cntb{j_0}{n}$ has Zimin index at most $n+1$.  
In particular, 
$z_2\tilde{y}\chi\tilde{x}$, which is a suffix, also has Zimin index at most $n+1$.

As $z_2\tilde{y}\chi\tilde{x}$ is equal to $\delta \chi \delta\gamma_1$
we have that $\delta \chi \delta$ is an infix of the word ${z_2\tilde{y}\chi\tilde{x}=}z_2\cntb{j_0}{n}$ of Zimin index at most $n+1$. 
By Fact~\ref{fact:ZiminType-inductive}, this implies that $\ZiminType{\delta} \leq n$ which concludes the case.
\medskip

\textbf{Case 6: $ |\tilde{x}| > |\delta\gamma|$} 

\medskip

\begin{tikzpicture}[scale=0.55]
\node at (-2,1) {$\alpha=$};
\draw (0,0) rectangle (6,1);
\node at (3,0.5) {$\delta$}; 
\draw (6,0) rectangle (12,1);
\node at (9,0.5) {$\gamma$}; 
\draw (12,0) rectangle (18,1);
\node at (15,0.5) {$\delta$}; 

\draw (0,1) rectangle (14,2);
\node at (7,1.5) {$\tilde{x}$};
\draw (14,1) rectangle (16,2);
\node at (15,1.5) {$\psi(b)$};
\draw (16,1) rectangle (18,2);
\node at (17,1.5) {$\tilde{y}$};

\end{tikzpicture}

This situation cannot occur under our hypothesis.
Indeed $\alpha$ would contain two occurrences of $\psi(b)$ which by Lemma~\ref{lemma:occurences-letter-parse} implies that the center of its parse contains two occurrences of the order $n+1$ symbol $b$. 
This brings
a contradiction to the fact that the context $c_1$ contains exactly one symbol of order $n+1$.

We have shown that for all $i\in[0,\tau(n+1)-1]$ we have
\[
\Zimin{\cntb{i}{n+1}} \leq n+2.
\]

Let us now show that for all $b \in \{0_{n+2},1_{n+2}\}$, we have
\begin{itemize}
\item $\Zimin{\psi(b) \cntb{i}{n+1}}\leq n+2$ and
\item $\Zimin{\cntb{i}{n+1}\psi(b)} \leq n+2$.
\end{itemize}

We first consider the case of $\psi(b) \cntb{i}{n+1}$.
Let $\alpha \beta \alpha$ be an infix of $\psi(b) \cntb{i}{n+1}$ for some non-empty $\alpha$ and $\beta$. By Fact~\ref{fact:ZiminType-inductive}, it is enough to show that $\ZiminType{\alpha}\leq n+1$. By Lemma~\ref{lem:Zimin-simple}, it is enough to consider the case when $\alpha$ is non-simple and hence by Lemma~\ref{lem:unique-parse}, $\alpha$ admits a 
{unique} parse $p=(\ell,u,r)$.

As $\alpha \beta \alpha$ is an infix of  $\psi(b) \cntb{i}{n+1}$, there exists $z_1$ and $z_2$ such that
\[
\psi(b) \cntb{i}{n+1} =z_1 \alpha \beta \alpha z_2.
\]

We distinguish different possible lengths of $z_1$.

\noindent
\textbf{Case 6A: $|z_1| \geq |\psi(b)|$.} 
In this case, $\alpha \beta \alpha$ is an infix of $\cntb{i}{n+1}$. 
We have already shown that $\Zimin{\cntb{i}{n+1}} \leq n+2$
and thus $\ZiminType{\alpha}\leq n+1$.
\medskip

\noindent
\textbf{Case 6B: $|z_1|=0$.}  We will show that $|\alpha|<|\psi(b)|$
and hence by Lemma~\ref{lemma:zimin-monotone} we have
$\Zimin{\alpha}{<|\psi(b)|=} \lfloor \log_2(2n+6) \rfloor \leq n+1$ as $n \geq 3$.

Assume towards a contradiction that $|\alpha| \geq |\psi(b)|$. Hence $\psi(b)$ is a prefix 
of $\alpha$. Therefore, $\psi(b)\cntb{i}{n+1}$ would contain two occurrences of $\psi(b)$. By Lemma~\ref{lemma:occurences-letter-parse}, $b\cnt{i}{n+1}$
would contain two occurrences of the order $n+2$ symbol $b${,} which brings the contradiction.
\medskip

\noindent
\textbf{Case 6C: $1\leq |z_1| < |\psi(b)|$.} We will show that $\alpha$ 
is an infix of $\psi(b') \cntb{0}{n}$ for some $b' \in \{0_{n+1},1_{n+1}\}$.
Note that this will be sufficient since then we can apply induction hypothesis
to conclude that $\ZiminType{\alpha}\leq n+1$.

As $1 \leq |z_1| < |\psi(b)|$, the parse $p=(\ell,u,r)$ is
such that $\ell$ is a {non-empty} suffix of $\psi(b)$ and $u$ is a prefix of $\cnt{i}{n+1}$. 
Let us first show that $\cnt{0}{n}$ is not a prefix of $u$.
Assume towards a contradiction
that $\cnt{0}{n}$, which is a prefix of $\cnt{i}{n+1}$, is also a prefix of $u$. 
By Lemma~\ref{lem:bijection_occurrences}, this would imply that $\cnt{i}{n+1}$ contains two occurrences of $\cnt{0}{n}$, which brings the contradiction. 
Thus, $u$ is not a prefix of $\cnt{0}{n}$ and hence
$\psi(u)r$ is a prefix of $\cntb{0}{n}$. 

It remains to show that $\ell$ is a suffix of $\psi(b')$ for some
$b' \in \{0_{n+1},1_{n+1}\}$. 
As there are two occurrences of the parse $p$ in $\psi(b)\cntb{i}{n+1}$, this 
implies that $\ell$ is the suffix of $\psi(b)$ and some $\psi(b'')$ for some
symbol $b''$ of order {$k\leq n+1$}.
From the definition of $\psi$, it follows that $\ell$ is a suffix of $(01)^{k-1}00$ or $(01)^{k-1}11$. Hence as announced, $\ell$ is a suffix 
of an order $n+1$ symbol.
\bigskip

\noindent
We have shown that $\Zimin{\psi(b) \cntb{i}{n+1}}\leq n+2$.

It remains to consider the case of $\cntb{i}{n+1}\psi(b)$. 
Remark that, as the definition of higher-order counters is not symmetrical with respect to left-right and 
right-left, this case is not identical to the previous one.

Let $\alpha \beta \alpha$ be an infix of $\cntb{i}{n+1}\psi(b)$ for some non-empty $\alpha$ and $\beta$. 
By Fact~\ref{fact:ZiminType-inductive}, it is enough to show that $\ZiminType{\alpha}\leq n+1$. 
By Lemma~\ref{lem:Zimin-simple}, it is enough to consider the case when $\alpha$ is non-simple and hence by Lemma~\ref{lem:unique-parse}, $\alpha$ has a unique parse $p=(\ell,u,r)$.

As $\alpha \beta \alpha$ is an infix of  $\cntb{i}{n+1} \psi(b)$, there exist $z_1$ and $z_2$ such that:
\[
\cntb{i}{n+1} \psi(b)  =z_1 \alpha \beta \alpha z_2
\]

We distinguish cases on the length of $z_2$.
\medskip

\noindent
\textbf{Case 6D: $|z_2| \geq |\psi(b)|$.} In this case, $\alpha \beta \alpha$ is an infix of $\cntb{i}{n+1}$. We have already shown that $\Zimin{\cntb{i}{n+1}} \leq n+2$.
\medskip

\noindent
\textbf{Case 6E: $|z_2|=0$.}  We will show that $|\alpha|<|\psi(b)|$
and hence by Fact~\ref{lemma:zimin-monotone}, $\Zimin{\alpha} \leq \lfloor \log_2(2n+6) \rfloor \leq n+1$ as $n \geq 3$.

Assume towards a contradiction that $|\alpha| \geq |\psi(b)|$. Hence $\psi(b)$ is a suffix
of $\alpha$. Therefore, $\cntb{i}{n+1}\psi(b)$ would contain two occurrences of $\psi(b)$. By Lemma~\ref{lemma:occurences-letter-parse}, $\cnt{i}{n+1}b$ would contain two occurrences of the order $n+2$ symbol $b$ which brings the contradiction.
\medskip

\noindent
\textbf{Case 6F: $1\leq|z_2| < |\psi(b)|$.}


Recall that  $\cntb{i}{n+1}$ ends with $\cntb{\max{n}-1}{n} \psi(b')$ for some $b' \in \{0_{n+1},1_{n+1}\}$.

We now distinguish cases on the length of $\alpha z_2$. 

$\quad$\textbf{Subcase : $|\alpha z_2| \leq |\psi(b')\psi(b)$|.} 

As $b$ is an order $n+2$ symbol and $b'$ an order $n+1$ symbol,
we have that $|\alpha|< 4 + 2n + 4 + 2(n+1)$. 
By Lemma~\ref{lemma:zimin-monotone}, it follows that $\Zimin{\alpha} \leq \lfloor \log_2(4n+10) \rfloor$. 
Furthermore as for all $n \geq 3$ we have
$\lfloor \log_2(4n+10) \rfloor \leq n+1$, we can conclude this subcase.

\medskip

$\quad$\textbf{Subcase : $|\psi(b')\psi(b)|<|\alpha z_2| \leq |\cntb{\max{n}-1}{n}\psi(b')\psi(b)$|.} \\
In this case, the parse $p=(\ell,u,r)$ of $\alpha$ is such that:
\begin{itemize}
\item $r$ is a non-empty prefix of $\psi(b)$,
\item $u$ ends with the order $n+1$ symbol $b'$.	
\end{itemize}
By Lemma~\ref{lem:bijection_occurrences}, the parse $p$ has two occurences in $\cnt{i}{n+1}b$. 
Hence it has an occurrence in $\cnt{i}{n+1}$. 
As $u$ ends with an order $n+1$ symbol and as any symbol of order $n+1$ can only be followed by a symbol of order $1$ in $\cnt{i}{n+1}$, we have that $r$ is a 
strict prefix of an order $1$ symbol. In particular $|r|<4$.

We have established that $\alpha$ is a suffix of ${\cntb{\max{n}-1}{n}} \psi(b') r$ with $|r|<4$. It remains to prove that $\ZiminType{\alpha} \leq n+1$.

Consider a decomposition of $\alpha$ as  $\delta \gamma \delta$ for some non-empty $\delta$ and $\gamma$. Assume towards a contradiction that  $|\delta|\geq |\psi(b')r|$. In this case, $\psi(b')$ is an infix of $\delta$ and hence $\alpha$ would have two occurrences of $\psi(b')$. By Lemma~\ref{lemma:occurences-letter-parse}, the 
{center of $\alpha$'s parse} would contain two order $n+1$ symbols which contradicts the fact that
$\alpha$ is a suffix of  ${\cntb{\max{n}-1}{n}} \psi(b') r$ which {has precisely
one occurrence of the code of one order $n+1$ symbol}.

Hence we have $|\delta| < |\psi(b')r| \leq 2n +7$. 
By Lemma~\ref{lemma:zimin-monotone}, $\Zimin{\delta} \leq \lfloor \log_2(2n+7) \rfloor$. As for all $n \geq 3$, it holds that  $\lfloor \log_2(2n+7) \rfloor \leq n$. We have shown that $\Zimin{\delta}\leq n$ and hence $\Zimin{\alpha}\leq n+1$.

\medskip

$\quad$\textbf{Subcase : $|\alpha z_2| > |\cntb{\max{n}-1}{n}\psi(b')\psi(b)$|.}\\
 This case cannot occur under our assumptions. 
 Indeed, this would imply that the center $u$ of the parse $p=(\ell,u,r)$ of $\alpha$ contains $\cnt{\max{n}-1}{n}$. As the parse $p$ has at least two occurrences in $\cnt{i}{n+1}b$, it would imply that  $\cnt{\max{n}-1}{n}$ has two occurrences in $\cnt{i}{n+1}${,} which contradicts Lemma~\ref{lemma:basics-counters}.
\end{proof}

%% file: abelian.tex
\newcommand{\defeq}{\stackrel{\text{def}}{=}}
\newcommand{\defequiv}{\stackrel{\text{def}}{\Leftrightarrow}}
\newcommand\restrict[1]{{|}_{{#1}}}

\section{Avoiding Zimin patterns in the abelian sense}\label{S Abelian}

Matching a pattern in the abelian sense is a weaker condition, where one only requires that
all infixes that are matching a pattern variable must have the same number of occurrences of each letter
(instead of being the same words).
Hence, for two words $x,y\in A^*$ we write $x\equiv y$ if
$|x|_a=|y|_a$ for all $a\in A$.
Let $\rho=\rho_1\cdots\rho_n$ be a pattern, where $\rho_i\in\X$ is a pattern variable
for all $i\in[k]$.
An {\em abelian factorization of a word $w\in A^*$ for the pattern $\rho$}
is a factorization 
$w=w_1\cdots w_n$ such that $w_i\not=\varepsilon$ for all $i\in[n]$ and
$\rho_i=\rho_j$ implies $w_i\equiv w_j$ for all $i,j\in[n]$.
A word $w\in A^*$ {\em matches pattern $\rho$ in the abelian sense} if
there is an abelian factorization of $w$ for $\rho$.
The definitions when a word encounters a pattern in the abelian sense and when a pattern is unvavoidable in the abelian sense
are as expected.

We note that every pattern that is unavoidable is in particular unavoidable in the abelian sense.
However, the converse does not hold in general as witnessed by the pattern $xyzxyxuxyxzyx$ as shown in \cite{CuLi01}.

To the best of the authors' knowledge abelian unavoidability still lacks a characterization in the style of general unavoidability in terms of Zimin patterns;
we refer to \cite{Currie05} for some open problems and conjectures.
Although being possibly less meaningful as for general unavoidability, the analogous Ramsey-like function for
abelian unavoidability has been studied.

\begin{definition}
Let $n,k\geq 1$. We define
$$
g(n,k)=\min\{\ell\geq 1\mid\forall w\in[k]^\ell:w\text{ encounters }Z_n\text{ in the abelian sense}\}.
$$
\end{definition}

Clearly, $g(n,k)\leq f(n,k)$ and to the best of the authors' knowledge no elementary upper bound has been shown for $g$ so far.
By applying a combination of the probabilistic method \cite{AS15} and 
analytic combinatorics \cite{FS09}
Tao showed the following lower bound for $g$.
\begin{theorem}[Tao \cite{Tao14}, Corollary 3\label{T Tao}]
Let $k\geq 4$. Then
$$g(n,k)\geq
 (1+o(1))\sqrt{2\prod_{j=1}^{n-1}\left[\sum_{\ell=1}^\infty\frac{1}{k^{2^j\ell}}\sum_{i_1+\cdots+i_k=\ell}
{\ell \choose i_1,\ldots,i_k}\right]^{-1}}\qquad.$$
\end{theorem}

Unfortunately, it was not clear to {the authors} what the asymptotic behavior of this lower bound is. However Jugé \cite{JugePC} provided us  with an estimate of its asymptotic behavior. 

\begin{corollary}[Jug\'e \cite{JugePC}]\label{cor:juge}
	Let $k \geq 4$. The expression in Theorem \ref{T Tao}, and hence
	$g(n,k)$, is lower-bounded by
	$$
	  \left(\dfrac{1}{\sqrt{21}}+o(1)\right) \dfrac{k^{2^{n-1}}}{{k}^{(n+1)/2}}\qquad.
	$$
\end{corollary}

\bigskip

In Section \ref{S abelian Lower} we prove another doubly-exponential lower bound on $g$ by applying
the first moment method  \cite{AS15}. 
Our lower bound on $g$ is not as good as the {one obtained by combining}
Theorem~\ref{T Tao} with Corollary~\ref{cor:juge} but its
proof seems more direct (already more direct than the proof of Theorem~\ref{T Tao}
itself).
The proof follows a similar strategy as the (slightly better) doubly-exponential 
lower bound for $f$ from \cite{CR14}, but again, seems to be more direct.
Our novel contribution is to provide a doubly-exponential upper bound on $g$ in Section \ref{S abelian Upper}. 
Note that Tao in \cite{Tao14} only provides a non-elementary 
upper bound for the non-abelian case. 

\subsection{A simple lower bound via the first-moment method\label{S abelian Lower}}
For all $n\geq 1$ let $\X_n=\{x_1,\ldots,x_n\}$ denote the set of the first
$n$ pattern variables. We note that the variable $x_i$ appears precisely $2^{n-i}$ times in 
$Z_n$ and its first occurrence is at position $2^{i-1}$ for all $i\in[1,n]$.
An {\em abelian occurrence of $Z_n$} in a word $w$ is a pair 
$(j,\lambda)\in [0,|w|-1]\times\N^{\X_n}$
for which there is an factorization $w=uvz$ with $|u|=j$
 and an abelian factorization $v_1\cdots v_{2^n-1}$ of $v$ for $Z_n$
satisfying $\lambda(x_i)=|v_{2^{i-1}}|$.
%

By applying the probabilistic method \cite{AS15} we show a lower bound 
for $g(n,k)$ that is doubly-exponential in $n$ for every fixed $k\geq 2$. The proof is similar the lower bound proof from \cite{CR14}.

\begin{theorem}
Let $k\geq 2$. Then
$$g(n,k)> 
k^{\left\lfloor\frac{2^{n}}{n+2}\right\rfloor-1}\quad.$$
\end{theorem}
\begin{proof}
For $n,\ell\geq 1$ let $\Delta_{n,k,\ell}$ denote the expected number of abelian occurrences of $Z_n$ in a 
random word in the set $[k]^\ell$. Remark that we always consider the uniform distribution over words.
If $\Delta_{n,k,\ell}<1$, then by the probabilistic method \cite{AS15} there exists a word of length $\ell$
over the alphabet $[k]$ that does {\em not} encounter $Z_n$ in the abelian sense; hence we 
can conclude $g(n,k)>\ell$.
Therefore we investigate those $\ell=\ell(n,k)$ for which we can guarantee $\Delta_{n,k,\ell}<1$.
We need two intermediate claims.

\medskip

\noindent
{\em Claim 1.} 
{The probability that $m$ pairwise independent random words $w_1,w_2\ldots, w_m$ in $[k]^h$ 
satisfy $w_1\equiv w_2 \equiv \cdots\equiv w_m$ is at most $(1/k)^{m-1}$.}

\medskip

\noindent 
{\em Proof of Claim 1.\ }
{We only show the claim only for $m=2$, the case when $m>2$ can be shown analogously.
Let $A_{k,h}$ denote the event that two independent random words $u$ and $v$ in $[k]^h$ satisfy $u\equiv v$. 
Then $\Pr(A_{k,h})\leq 1/k$ for all $h\geq 1$.  }
For every word 
$w=w_1\cdots w_h\in[k]^h$, let $\oplus_k w=\left(\sum_{i=1}^h w_i\right)\text{ mod }k$. 
Remark that $u \equiv v$ implies that $\oplus_k u = \oplus_k v$.
Let us fix any $j\in[k]$. Then we clearly have
$\Pr\left[\oplus w=j\right]=1/k$ for every random word
$w=w_1\cdots w_h\in[k]^h$.
Thus,
$$
\begin{array}{lcl}
\Pr(A_h) &\leq
 & \sum_{j\in[1,k]}^k \Pr[\oplus_k u=\oplus_k v=j]\\
 &=&
\sum_{j\in[1,k]}^k \Pr[\oplus_k u=j]\Pr[\oplus_k v=j]=
\sum_{j\in[1,k]}^k 1/k^2=
1/k
\end{array}
$$
\noindent
{\em End of the proof of Claim 1.\ }\\

\newcommand{\width}{\text{width}}

Recall that $Z_n=y_1\cdots y_{2^n-1}$, where $y_i\in\{x_1,\ldots,x_n\}$ for all $i\in[2^n-1]$
and that the variable $x_i$ appears precisely $2^{n-i}$ times in $Z_n$.
We recall that we would like to bound the expected number of occurrences (in the abelian sense) of $Z_n$ in a random 
word of length $\ell$ over the alphabet $[k]$.
To account for this, we define for each mapping $\lambda:\X_n\rightarrow\N^+$ its
{\em width} as $\width(\lambda)=\sum_{i=1}^n2^{n-i}\cdot\lambda(x_i)$.
For every word $v$ of length $\width(\lambda)$ its (unique) {\em decomposition with respect to $\lambda$}
is the unique factorization $v=v_1\cdots v_{2^n-1}$ such that $y_j=x_i$ implies
$|v_j|=\lambda(x_i)$ for all $j\in[2^n-1]$ and all $i\in[n]$.

\medskip

\noindent
{\em Claim 2.} Let $\lambda:\X_n\rightarrow\N^+$ and let $B_\lambda$ denote the event that 
{in a random word from $[k]^d$ we have that} $(0,\lambda)$ is an occurrence of $Z_n$ in the abelian sense.
Then $\Pr(B_\lambda)\leq k^{n-2^n+1}$. 

\medskip

\noindent
{\em Proof of Claim 2.\ } Let $\lambda:\X_n\rightarrow\N^+$ with $d=\width(\lambda)$. For $i \in [n]$, let $j_1^{(i)}<\cdots<j_{2^{n-i}}^{(i)}$ be an enumeration of the $2^{n-i}$ indices corresponding to occurrences of $x_i$ in $Z_n$. 
For all $i\in[n]$ consider the event $B_\lambda^{(i)}$ that  a random word of $[k]^d$ has its decomposition with respect to $\lambda$ of
the form $v_1 \cdots v_{2^n-1}$ such that the words $v_{j_1^{(i)}},\ldots,v_{j_{2^{n-i}}^{(i)}}$ (which are all of length $\lambda(x_i)$)
are pairwise equivalent with respect to $\equiv$.
The event $B_{\lambda}$ is the intersection of the events $B_{\lambda}^{(1)},\ldots,B_{\lambda}^{(n-1)}$ and $B_{\lambda}^{(n)}$. As $B_{\lambda}^{(1)},\ldots,B_{\lambda}^{(n-1)}$ and $B_{\lambda}^{(n)}$ are mutually independent events, the probability $\Pr(B_\lambda)$ is equal to $\prod_{i=1}^{n} \Pr(B_{\lambda}^{(i)})$.
%
%
%
%
%
\noindent
We have 
\begin{eqnarray}
\Pr(B_{\lambda})= \prod_{i=1}^n \Pr(B_\lambda^{(i)}) \stackrel{\text{Claim 1}}{\leq} \prod_{i=1}^n (1/k)^{2^{n-i}-1}=
k^{-\left(\sum_{i=1}^n2^{n-i}\right)+n}=
k^{n-2^n+1}.
\end{eqnarray}
{\em End of the proof of Claim 2.}\\

It is clear that for every $(j,\lambda)$, where $d=\width(\lambda)$ and $j+d\leq\ell$,
the probability that $(j,\lambda)$ is an occurrence of a random word from $[k]^\ell$
equals to probability that $(0,\lambda)$ is such an occurrence
and therefore equals $\Pr(B_\lambda)$. 
Thus, this probability does not depend on $j$.

We are ready to prove an an upper bound for $\Delta_{n,k,\ell}$, where 
we note that any occurrence $(j,\lambda)$ of $Z_n$ in a random word of length
$\ell$ must satisfy $\width(\lambda)\geq 2^n-1$.

\begin{eqnarray}
\Delta_{n,k,\ell} & \leq &
\sum_{d=2^n-1}^\ell\sum_{j=0}^{\ell-d}
\sum_{\lambda:\X_n\rightarrow\N^+\atop
\width(\lambda)=d}\Pr\left[(j,\lambda)\text{ is an occ. in a random word in $[k]^\ell$}\right]\nonumber\\
&\leq&\sum_{d=2^n-1}^\ell\sum_{j=0}^{\ell-d}
\sum_{\lambda:\X_n\rightarrow\N^+\atop
\width(\lambda)=d}\ \Pr(B_\lambda)\nonumber\\
\nonumber\\
&\stackrel{\text{Claim 2}}{\leq}&
\sum_{d=2^n-1}^\ell\sum_{j=0}^{\ell-d}\sum_{\lambda:\X_n\rightarrow\N^+\atop
\width(\lambda)=d}\ k^{n-2^n+1}\nonumber\\
&\leq&
\sum_{d=2^n-1}^\ell\sum_{j=0}^{\ell-d}\quad
d^n\ \cdot\ k^{n-2^n+1}\nonumber\\
&\leq&
\sum_{d=2^n-1}^\ell
\ell\ \cdot\ d^n\ \cdot\ k^{n-2^n+1}\nonumber\\
&\leq&
\frac{\ell^2\ \cdot\ \ell^n}{k^{2^n-n-1}}\nonumber\\
&=&
\frac{\ell^{n+2}}{k^{2^n-n-1}}
\label{Eq}
\end{eqnarray}

\noindent
We finally determine \modified{the largest value of $\ell$ that still guarantees that}  $\Delta_{n,k,\ell}<1$.
\begin{eqnarray*}
\Delta_{n,k,\ell}<1
&\quad\stackrel{(\ref{Eq})}{\Longleftarrow}\quad&
\frac{\ell^{n+2}}{k^{2^n-n-1}}<1\\
&\quad\Longleftarrow\quad&
\ell^{n+2}<k^{2^n-n-1}\\
&\Longleftarrow&
(n+2)\log_k \ell<2^n-n-1\\
&\Longleftarrow&
\log_k \ell<\frac{2^n-n-1}{n+2}\\
&\Longleftarrow&
\ell<k^{\frac{2^n-n-1}{n+2}}\\
&\Longleftarrow&
\ell<k^{\frac{2^n}{n+2}-\frac{n+1}{n+2}}\\
&\Longleftarrow&
\ell=k^{\left\lfloor\frac{2^{n}}{n+2}\right\rfloor-1}
\end{eqnarray*}
\end{proof}

\subsection{A doubly-exponential upper bound\label{S abelian Upper}} 

Let us finally prove an upper bound for $g(n,k)$ that is doubly-exponential in $n$.
\begin{theorem}
$g(n,k)\leq 2^{(4k)^{n}(n-1)!}$.
\end{theorem}
\begin{proof}
We prove the statement by induction on $n$.
For $n=1$ we have
$$
g(1,k)\quad=\quad1\quad\leq\quad 2^{(4k)^1(1-1)!}\quad.
$$
For the induction step, let $n\geq 1$ and let us assume induction hypothesis for $g(n,k)$.
To determine an upper bound $g(n+1,k)$ we consider any sufficiently long word $w\in[k]^+$ that we 
can factorize as $w=w_1a_1w_2a_2\cdots w_m a_m z$, where $|w_j|=g(n,k)${,}
$a_j\in[k]$ for all $j\in[m]$ and $z\in[k]^*$, where $m$ is assumed sufficiently large 
for the following arguments to work.
By induction hypothesis for all $j\in[m]$, $w_j$ encounters $Z_n$ in the abelian sense, 
witnessed in some infix $v_j$ and some abelian factorization
$v_j=v_j^{(1)}\cdots v_j^{(2^n-1)}$ for $Z_n$ .
To each such abelian factorization we can assign the Parikh image how the 
word $v_j$ matches each variable $x_i$ (with $i\in[n]$) that appears in $Z_n$.
Formally, each of the above abelian factorizations $v_j=v_j^{(1)}\cdots v_j^{(2^n-1)}$ induces a mapping
$\psi_j:\X_n\rightarrow\N^{[k]}$ such that
${\psi_j}(x_i)(t)=|v_j^{(2^i-1)}|_t$ for all $j\in[m]$, all $i\in[n]$ and all $t\in[k]$.
As expected, we write $\psi_j\equiv\psi_h$ if $\psi_j(x_i)=\psi_j(x_i)$ for all $i\in[n]$.
Note that if there are distinct $j,h\in[1,m]$ with $\psi_j\equiv\psi_h$, then clearly
$w$ encounters $Z_{n+1}=Z_nx_{n+1}Z_n$ in the abelian sense.
Let us therefore estimate a sufficiently large bound on $m$ such that there are always
 two distinct indices $i,j\in[1,m]$ that satisfy $\psi_i\equiv\psi_j$.

It is easy to see that there are at most $g(n,k)^{kn}$ different equivalence classes
for the $\psi_j$ with respect to $\equiv$. 

Therefore by setting $m= g(n,k)^{kn}+1$ we have shown
\begin{eqnarray}
g(n+1,k)\quad\leq\quad (g(n,k)+1)(g(n,k)^{kn}+1)\quad.\label{E Recurrence}
\end{eqnarray}
Hence, we obtain
\begin{eqnarray*}
g(n+1,k)&\quad\stackrel{(\ref{E Recurrence})}{\leq}\quad& (g(n,k)+1)(g(n,k)^{kn}+1)\\
&\stackrel{g(n,k)\geq 1}{\leq}& 2 \cdot g(n,k) \cdot 2\cdot g(n,k)^{kn}\\
&=&4\cdot g(n,k)^{kn+1}\\
&\stackrel{n\geq 1}{\leq}&
4\cdot g(n,k)^{2kn}\\
&\stackrel{\text{IH}}{\leq}&
4\cdot \left(2^{(4k)^n(n-1)!}\right)^{2kn}\\
&=&
4\cdot 2^{2\cdot 4^n k^{n+1}n!}\\
&=&
2^{2\cdot 4^nk^{n+1}n!+2}\\
&\stackrel{n\geq 1}{\leq}&
2^{2(2\cdot 4^nk^{n+1}n!)}\\
&=&
2^{(4k)^{n+1}n!}
\end{eqnarray*}

\end{proof}

%% file: conclusion.tex
\section{Conclusion}\label{S Conclusion}

We have established a lower bound for $f(n,k)$ that is already non-elementary when $k=2$.
A first element of an answer is that the first moment method used in \cite{CR14} cannot be used to obtain a lower bound that is asymptotically above doubly-exponential. 
Indeed, as for a length \modified{$\ell \geq  k^{2^n-n-1}+2^n$},
the expected number $\Delta_{n,k,\ell}$ of occurrences $Z_n$ in 
a random word in $[k]^\ell$ is greater than $1$. 

To see this, recall that $|Z_n|=2^n-1$ and hence there is at most one possible occurrence of
$Z_n$ in any word of length $2^n-1$. Let $A_n$ denote the event that $Z_n$ is encountered in a random word in $[k]^{2^n-1}$.
We have
$$
\modified{\Pr(A_n)=\prod_{i=1}^n (1/k)^{2^{n-i}-1}=k^{-2^n+n+1}.}
$$

Assume that $\ell \geq k^{2^n-n-1}+2^n$. For each $i \in [0,k^{2^n-n-1}]$, let $X_i$ be the indicator random variable marking that the 
infix, of a random word in $[k]^\ell$, occurring at $i$ and of length $2^n-1$ matches $Z_n$. By linearity of the expectation, it follows that
\[
\Delta_{n,k,\ell} \geq \sum_{i=0}^{k^{2^n-n-1}} E(X_i) \geq  (k^{2^n-n-1}+1) \Pr(A_n) = 1 + \dfrac{1}{k^{2^n-n-1}} \geq 1.
\]

{Thus, more advanced probabilistic method techniques are necessary. Indeed, very recently~\cite{Colon17} Condon, Fox and Sudakov  have applied the local lemma to obtain non-elementary lower bounds on $f(n,k)$. 

For the abelian case, an explicit family of words witnessing the doubly-exponential lower bound seems worth investigating. 
}